\documentclass[10pt,conference]{IEEEtran}

\usepackage{numprint}
\npdecimalsign{\ensuremath{.}}
\newcommand{\np}[1]{\numprint{#1}}
\usepackage[scientific-notation=true]{siunitx}
\newcommand{\sn}[1]{\num[output-exponent-marker = \ensuremath{\mathrm{E}},retain-zero-exponent = true]{#1}}
\newcommand{\snt}[1]{\newcolumntype{1}{{Decimal-centred}}{\sn{#1}}}
\newcommand{\worst}[1]{\bfseries{#1}}
\newcommand{\best}[1]{{\bfseries#1}}

\newcommand{\param}[1]{k_{#1}}
\newcommand{\paramrev}[1]{\overline{k}_{#1}}

\usepackage{amsmath}
\usepackage{amsthm}
\usepackage{amsfonts}
\usepackage{amssymb}
\usepackage{times}
\usepackage{dsfont}
\usepackage{balance}

\usepackage{hyperref}

\usepackage{microtype}
\usepackage{extarrows}
\usepackage{algorithm}
\usepackage{clrscode3e}
\usepackage{subfig}

\usepackage{booktabs}
\usepackage{todonotes}

\usepackage{tikz}
\usetikzlibrary{arrows,decorations.pathmorphing} 

\usepackage{multirow}
\usepackage{colortbl}
\usepackage[inline]{enumitem}

\newtheorem{lemma}{Lemma}
\newtheorem{proposition}{Proposition}
\newtheorem{theorem}{Theorem}
\newtheorem{remark}{Remark}
\newtheorem{corollary}{Corollary}
\newtheorem{definition}{Definition}
\newtheorem{examplenum}{Example}

\newcommand{\cf}{\textit{cf.}}

\newcommand{\RE}{\mathbb{R}}
\newcommand{\dom}{\mathit{dom}}
\newcommand{\mon}{\mathbf{M}}
\newcommand{\lin}{\mathbf{L}}
\newcommand{\poly}{\mathbf{P}}

\newcommand{\calF}{\mathcal{F}}

\newcommand{\R}{R}
\newcommand{\Q}{Q}

\newcommand{\X}{{X}}

\newcommand{\supp}{\mathit{supp}}
\newcommand{\gfp}[1]{\mathtt{gfp}(#1)}
\newcommand{\jacdec}{\proc{JacDec}}
\newcommand{\onthefly}{\proc{FindBDB}} 

\newcommand{\refnode}[3][]{
  \tikz[remember picture, baseline, anchor=base #1]
  \node[inner sep=0ex, fill=lightgray!50!white,  #1] (#3) {#2}; %
}

\usepackage{listings}
\newcommand{\ch}[1]{{{#1}}}

\begin{document}
\newcolumntype{H}{>{\setbox0=\hbox\bgroup}c<{\egroup}@{}}
\title{\ch{Efficient} Local Computation of Differential Bisimulations via Coupling and Up-to Methods}

\author{\IEEEauthorblockN{Giorgio Bacci\IEEEauthorrefmark{1},
Giovanni Bacci\IEEEauthorrefmark{1},
Kim G. Larsen\IEEEauthorrefmark{1},
Mirco Tribastone\IEEEauthorrefmark{2},
Max Tschaikowski\IEEEauthorrefmark{1}, and
Andrea Vandin\IEEEauthorrefmark{3}}
\IEEEauthorblockA{\IEEEauthorrefmark{1}Department of Computer Science, Aalborg University, Denmark}
\IEEEauthorblockA{\IEEEauthorrefmark{2}SysMA Research Unit, IMT Lucca, Italy}
\IEEEauthorblockA{\IEEEauthorrefmark{3}Institute of Economics, Scuola Superiore Sant'Anna, Italy}}

\IEEEoverridecommandlockouts
\IEEEpubid{\makebox[\columnwidth]{978-1-6654-4895-6/21/\$31.00~
\copyright2021 IEEE \hfill} \hspace{\columnsep}\makebox[\columnwidth]{   }}

\maketitle

\begin{abstract}
We introduce polynomial couplings, a generalization of probabilistic couplings, to develop an algorithm for the computation of equivalence relations which can be interpreted as a lifting of probabilistic bisimulation to polynomial differential equations, a ubiquitous model of dynamical systems across science and engineering. The algorithm enjoys polynomial time complexity and complements classical partition-refinement approaches because:
\begin{enumerate*}[label=(\alph*)]
\item it implements a local exploration of the system, possibly yielding equivalences \ch{that do not necessarily involve the inspection of the whole system of differential equations};
\item it can be enhanced by up-to techniques; and
\item it allows the specification of pairs which ought not be included in the output.
\end{enumerate*}
\ch{Using a prototype, these advantages are demonstrated on case studies from systems biology for applications to model reduction and comparison.
Notably, we report four orders of magnitude smaller runtimes than  partition-refinement approaches when disproving equivalences between Markov chains.}
\end{abstract}

\section{Introduction} \label{sec:intro}

Ordinary differential equations (ODEs) are a fundamental tool for modelling systems with continuous-time dynamics across science and engineering. In computer science, ODEs are central for the quantitative analysis of systems. For instance, in a continuous-time Markov chain (CTMC) a system of linear ODEs gives the forward equations of motion of the probability distribution~\cite{DBLP:conf/concur/BortolussiH12}; ODEs also serve as the underlying semantics of formal languages based on process algebra, with application to the analysis of distributed computing systems~\cite{Hillston_QEST05} or computational systems biology~\cite{biopepa}. 

Analogously to the classical non-deterministic setting based on labeled transition systems, bisimulations for ODEs have been proposed for the related purposes of model \emph{comparison} and model \emph{reduction}~\cite{lics16,popl16,DBLP:conf/hybrid/Boreale18,DBLP:journals/scp/Boreale20}, with applications that have transcended computer science (e.g.,~\cite{DBLP:journals/tac/PappasLS00,bisimulation_lin_sys_Schaft}). ``Lumping'' refers to a class of methods to reduce a system of ODEs onto lower-dimensional space such that each variable in the reduced ODE system represents an appropriate mapping of the set of original variables~\cite{LiRabitz1997,okino1998}.

An equivalence relation over the variables of an ODE system can be seen as a specific type of lumping. Indeed, this is well-known for CTMCs, \ch{where \emph{backward} and \emph{forward bisimulations}} (e.g., \cite{Larsen19911,1703385,BuchholzOrdinaryExact})  can be computed using lumping algorithms based on partition refinement~\cite{DBLP:journals/ipl/DerisaviHS03,DBLP:conf/tacas/ValmariF10}. More recently, partition refinement algorithms have been provided for a class of nonlinear ODEs, by means of symbolic approaches based on satisfiability modulo theories~\cite{popl16}, and polynomial ODEs,  generalizing bisimulation relations and related  lumping algorithms for CTMCs~\cite{pnas17}.

Partition refinement is efficient in finding the largest equivalence, i.e., the coarsest aggregation of an ODE system. However, \ch{it performs a \emph{global} exploration of the state space, i.e., it requires the availability of the whole system of ODEs}. Moreover, there are certain applications for which  alternative approaches may be more desirable. For example, when comparing two models, it would suffice to know whether there exists \emph{some} bisimulation relating the pair of initial states. For large-scale models in particular, it would be useful to prove (or disprove) this by computing smaller relations that do not necessarily involve exploring the whole state space.

Another question where partition refinement may not be appropriate regards the quest for a bisimulation that \emph{does not} contain some pairs. This is motivated by the fact that relating two variables may impose certain pre-conditions on the quotient model that domain-specific knowledge must exclude. For instance, in a backward bisimulation for CTMCs (related to the notion of exact lumpability~\cite{BuchholzOrdinaryExact}), related states must be initialized with the same initial probability; however, this may not always be meaningful from a modeling viewpoint. For example, variables that represent distinct discrete states of some model components (e.g., genes which can be active/inactive) must start from independent initial conditions.
\ch{In applications to dynamic models of regulatory networks, interesting research findings regard
genes 
that respond identically to different external stimuli, a.k.a.\ \emph{gene co-expression networks} (e.g.,~\cite{tlgl,ZhangH05}). To formally prove such a fact in the model, one would wish to find a bisimulation relation that relates the variables representing the genes of interest, but not those encoding the external stimuli.}
The problem with such kinds of \emph{negative constraints} is intrinsic to the refinement technique. Indeed, it works by splitting candidate equivalence classes into finer partitions; crucially, however, the modeler is left with the choice of the initial guess among the possibly exponentially many satisfying given negative constraints.

In this paper we propose a new approach that aims to tackle the above issues by presenting an algorithm for the construction of bisimulations \ch{based on a local exploration of the model}, 
starting with a candidate relation containing the pairs of ODE variables to be proved equivalent as well as a set of constraints containing pairs that are not allowed in the output.
We consider ODE systems with polynomial derivatives. This is an important class of ODEs on its own---e.g., it subsumes models with the well-known mass-action kinetics arising in many natural sciences~\cite{Voit:2013aa}. Additionally, ODEs with other  forms of nonlinearity (e.g., trigonometric functions, exponentials, rational expressions) can be algorithmically translated into polynomial ODEs~\cite{DBLP:conf/fm/0009ZZZ15}.

We study bisimulations for two equivalences, and related partition-refinement algorithms, proposed for polynomial ODEs: \emph{backward differential equivalence} (BDE) and \emph{forward differential equivalence} (FDE)~\cite{popl16,pnas17}. The former relates ODE variables whose solutions are equal when started from identical initial conditions. The latter, instead, guarantees that the solution of the induced aggregated ODE corresponds to the exact sum of the ODE variables in each equivalence class. Notably, when restricted to a CTMC, BDE and FDE correspond to  CTMC backward and forward bisimulation, respectively. \ch{Therefore, our results for polynomial ODEs yield a local algorithm for CTMC bisimulations as a special case.}

Our approach for the comparison of polynomials fundamentally exploits (a variant of) the \emph{coupling method}, a standard tool in probability theory~\cite{LindvallBook,Thorisson95}. In our setting, we adapt the notion of probabilistic couplings to that of \emph{linear} and \emph{monomial couplings}. Intuitively, a coupling is a pairing of equivalent variables that respects the sum of the coefficients in the polynomials. Proving the existence of such a coupling is equivalent to proving the equivalence of the polynomials.

The benefit in using them is that couplings can be computed efficiently by solving a transportation problem~\cite{Dantzig51,FordF56} in strongly polynomial time (e.g., by employing Orlin's algorithm~\cite{Olrin88}). This allows
us to design a local algorithm that runs in polynomial time in the number of monomials present in the
polynomial ODE system.

For further computational improvement, we enhance the local algorithm by employing \emph{coinduction up-to techniques}. Since their introduction~\cite{milnerBook}, coinduction up-to techniques were proved useful in numerous proofs about concurrent systems (see~\cite{PousS12} for a list of references), found applications in abstract interpretation~\cite{BonchiGGP18}, and more recently, to improve standard automata algorithms~\cite{BonchiP13,Bonchi0K17}.

Using a prototype, we apply our algorithm to a number of case studies from different domains (chemical reaction networks, gene regulatory networks, CTMCs, and epidemiological models) to show:
\begin{enumerate*}[label=(\roman*)]
\item how our algorithm allows the natural encoding of domain-relevant queries that would otherwise require tedious trial-and-error with partition refinement;
\item how our algorithm allows for a local 
analysis of the polynomial ODE system possibly avoiding a full exploration of the state space; and
\item how the enhancement of up-to-techniques leads to faster runtimes.
\end{enumerate*}
When applied to the problem of checking equivalence between CTMCs, motivated by fundamental research questions in evolutionary biology (e.g.,~\cite{CardelliCRN}), our local algorithm outperforms the partition-refinement algorithms for BDE and FDE by up to four orders of magnitude.

\paragraph*{Further Related Work}
\ch{The closest line of research to this paper is by Boreale~\cite{DBLP:journals/lmcs/Boreale19,DBLP:journals/scp/Boreale20}, who has introduced the notion $\mathcal{L}$-bisimulation for polynomial ODEs 
and an algorithm for computing it using up-to context techniques. $\mathcal{L}$-bisimulation generalizes BDE because it can prove more general invariants. However, it is not comparable to FDE~\cite{DBLP:journals/lmcs/Boreale19}. 
In addition, the procedure for computing $\mathcal{L}$-bisimulations exploits Buchberger's algorithm~\cite{Buchberger76}, which has doubly exponential time complexity.}

Polynomial invariants of dynamical systems have been also studied in the field of verification where model minimization is not sought per se. For instance,~\cite{DBLP:conf/tacas/GhorbalP14} computes these by maximizing the dimension of the kernel of a symbolic matrix. As further works we mention~\cite{DBLP:conf/popl/SankaranarayananSM04,DBLP:conf/lics/Platzer12} and refer to~\cite{DBLP:journals/lmcs/Boreale19} for a detailed discussion.

Probabilistic couplings have been used to characterize~\cite{DBLP:conf/lics/JonssonL91} and provide efficient methods~\cite{DBLP:conf/cav/Baier96} to test probabilistic bisimu\-lation for discrete-time Markov chains~\cite{Larsen19911}. More recently, they were instrumental for devising efficient methods for computing bisimilarity distances~\cite{DBLP:conf/fossacs/ChenBW12,Bacci:2013}.
Coupling methods, are successfully employed for proving invariants in probabilistic programs~\cite{BartheEGHSS15,BartheGHS17}, with  applications in formal program verification~\cite{AguirreBBBG18,BartheEGHS18}, security~\cite{BartheFGGHS16,BalleBG18}, and 
randomized algorithms~\cite{BartheEGHS17}. To our knowledge, this is the first time that coupling methods are employed in the analysis of ODEs.

\paragraph*{Synopsis}
Section~\ref{sec:prelim} recalls the concepts of BDE and FDE and fixes notation. Section~\ref{sec:couplingmethod} introduces linear and monomial couplings and shows how these are used to compare polynomials; Section~\ref{sec:BDBandFBD} provides fixed point characterizations of BDE and FDE using coupling methods. Section~\ref{sec:onthefly} presents the local algorithm for computing differential bisimulations and discusses how to further enhance its performance by integrating up-to techniques. In Section~\ref{sec:applications} we report experimental results on a number of applications highlighting the flexibility of our method as opposed to partition refinement. 


\section{Preliminaries} \label{sec:prelim}

%

\paragraph*{Notation} Fix a finite set of variables $\X$ which will appear in polynomial ODE systems. For a binary relation $\R \subseteq \X \times \X$, we denote by $r(\R)$, $s(\R)$, $t(\R)$, and $e(\R)$ respectively the reflexive, symmetric, transitive, and equivalence closure of $\R$.

We denote by $\mon[\X]$ the set of monomials over $\X$, ranged over by $m, n, \ldots \in \mon[\X]$; by $\lin[\X]$ the set of linear combinations over $X$, ranged over by $g,h, \dots \in \lin[\X]$; and by $\poly[\X]$ the set of polynomials over $\X$, ranged over by $p, q, \ldots \in \poly[\X]$. Throughout the paper, we will assume that every polynomial is expressed as a linear combination of monomials, i.e., $p = \sum_{m_i \in I} \alpha_{m_i} m_i$ where $\alpha_{m_i} \in \RE$.\footnote{This assumption is not only for convenience but will allow for an unambiguous definition of the size of a polynomial vector field.} 

A linear combination $g \in \lin[\X]$ can be decomposed into its \emph{positive} and \emph{negative} parts so that $g = g^{+} - g^{-}$, where $g^{+}  = \sum_{x_i \in I} \alpha_{x_i} {x_i}$ and $g^{-}  = \sum_{x_j \in J} \alpha_{x_j} {x_j}$ for positive coefficients $\alpha_{x_l} > 0$. This decomposition is unique once we require $I$ and $J$ to be disjoint. For a monomial $m = \prod_{x_i \in I} x_i^{a_i}$ we write $m(x_i)$ for the exponent $a_i \in \mathbb{N}$ associated with the variable $x_i$, whereas for a linear combination $g = \sum_{x_i \in I} \alpha_{x_i} {x_i}$ we write $g(x_i)$ for the coefficient $\alpha_{x_i}$ associated with the variable $x_i$. For a polynomial $p = \sum_{m_i \in I} \alpha_{m_i} m_i$, instead, $p(m_i)$ denotes the coefficient $\alpha_{m_i}$ associated with the monomial $m_i$. When $\X$ is clear from the context, we will write $\mon$, $\lin$, and $\poly$ respectively for $\mon[\X]$, $\lin[\X]$, and $\poly[\X]$.

\paragraph*{Polynomial vector fields} A \emph{vector field} over $X$ is a map $f \colon \RE^\X \to \RE^\X$ that is totally differentiable. It is called \emph{polynomial} when, for all $x \in \X$, $f_x$ is a polynomial over $\X$. Given a polynomial vector field $f$, we write $f(v)$ for the evaluation of $f$ at $v \in \RE^\X$. For an initial condition $v(0) \in \mathbb{R}^\X$, Picard-Lindel\"{o}f's theorem ensures that the ODE system $\partial_t{v}(t) = f(v(t))$ induced by $f$ has a unique solution $v \colon \dom(v) \to \RE^\X$, $t \mapsto v(t)$, where $\partial_t$ denotes derivative with respect to time.

Next we introduce an example of a polynomial vector field arising from a simple application from systems biology which will be used throughout the paper for illustration.
\begin{examplenum}\label{example:001}
Consider a chemical reaction network (CRN) where an enzyme $B$ can bind forming a complex with a substrate $A$ through two independent binding sites according to the following reversible reactions R1, $\ldots$, R4:
\begin{align*}
\text{R1}\colon\quad & A_{00} + B  \underset{3}{\stackrel{2}{\rightleftharpoons}} A_{10} &
\text{R2}\colon\quad & A_{00} + B  \underset{3}{\stackrel{2}{\rightleftharpoons}} A_{01} \\
\text{R3}\colon\quad & A_{10} + B \underset{3}{\stackrel{1}{\rightleftharpoons}} A_{11} &
\text{R4}\colon\quad &A_{01} + B  \underset{3}{\stackrel{1}{\rightleftharpoons}} A_{11}
\end{align*}
The subscripts $i,j$ in chemical species $A_{ij}$ denote the availability of either binding site in the substrate $A$. Reactions R1 and R2 model reversible binding at either site. Reactions R3 and R4 model the case when only one site is available for binding. The value on each arrow indicates the \emph{kinetic rate parameter} for the reaction. By mass-action kinetics~\cite{Voit:2013aa}, 
the above CRN gives rise to the ODE system
\begin{align*}
\partial_t v_{A_{00}}(t) & = f_{A_{00}}(v(t)) , &
\partial_t v_{A_{01}}(t) & = f_{A_{01}}(v(t)) , \\
\partial_t v_{A_{10}}(t) & = f_{A_{10}}(v(t)) , &
\partial_t v_{A_{11}}(t) & = f_{A_{11}}(v(t)) , \\
\partial_t v_{B}(t) & = f_{B}(v(t))
\end{align*}
induced by the following polynomial vector field
\begin{align}\label{eq:running:example}
f_{A_{00}} & = -4 {A_{00}} B + 3 {A_{10}} + 3 {A_{01}} \nonumber \\
f_{A_{01}} & = 2 {A_{00}} B - 3 {A_{01}} - {A_{01}} B + 3 {A_{11}} \nonumber \\
f_{A_{10}} & = 2 {A_{00}} B - 3 {A_{10}} - {A_{10}} B + 3 {A_{11}} \\
f_{A_{11}} & = {A_{10}} B + {A_{01}} B - 6 {A_{11}} \nonumber \\
f_B & = -4 {A_{00}} B + 3 {A_{10}} + 3 {A_{01}} - {A_{10}} B - {A_{01}} B + 6 {A_{11}} \nonumber
\end{align}
with variables $\X = \{A_{00},A_{01},A_{10},A_{11},B\}$.
\end{examplenum}

\paragraph*{Backward differential equivalence} We recall the notion of backward differential equivalence (BDE) from~\cite{popl16}.
\begin{definition}[Backward differential equivalence]\label{def:BDE}
Let $f$ be a vector field over $\X$. An equivalence relation $\R \subseteq \X \times \X$ is a BDE for $f$ if the implication
\begin{equation*}
\Big( \bigwedge_{(x,y) \in \R} v_x = v_y \Big) \Rightarrow \Big( \bigwedge_{(x,y) \in \R} f_x(v) = f_y(v) \Big)
\end{equation*}
is true for all $v \in \RE^\X$.
\end{definition}


\begin{examplenum}\label{example:002}
Consider the vector field $f$ from Example~\ref{example:001}. Then, the equivalence relation
\[
\R = \mathtt{id} \cup \{(A_{01},A_{10}),(A_{10},A_{01})\} \, ,
\]
where $\mathtt{id} = \{(x,x) \mid x \in \X\}$ denotes the identity relation, is a BDE for $f$ given in~(\ref{eq:running:example}).
\end{examplenum}

A BDE relates variables with identical ODE solutions when initialized equally~\cite[Theorem~3]{popl16}.
This property allows one to reason about the solutions of the ODE system induced by $f$ by looking at the smaller ODE system induced by the vector field $\hat{f} \colon  \RE^{\X/\R} \to \RE^{\X/\R}$
obtained by a change variable $H = x$, for each $H \in \X/\R$ and $x \in H$, mapping each variable to its equivalence class.


\begin{examplenum}\label{example:0022}
The equivalence classes of $\R$ from Example~\ref{example:002} are
\begin{equation*}
H_1 = \{ A_{00} \} \,,
H_2 = \{ A_{01}, A_{10} \} \,,
H_3 = \{ A_{11} \} \,,
H_4 = \{ B \} \,.
\end{equation*}
Then, the BDE-reduced vector field of $f$ given in~\eqref{eq:running:example} is
\begin{align*}
\hat{f}_{H_1} & = -4 {H_1} {H_4} + 6 {H_2} \\
\hat{f}_{H_2} & = 2 {H_1} {H_4} - 3 {H_2} - {H_2} {H_4} + 3 {H_3} \\
\hat{f}_{H_3} & = {H_2} {H_4} + {H_2} {H_4} - 6 {H_3} \\
\hat{f}_{H_4} & = -4 {H_1} {H_4} + 6 {H_2} - {H_2} {H_4} - {H_2} {H_4} + 6 {H_3}
\end{align*}
Let $v$ and $\hat{v}$ denote the solutions of the ODEs induced respectively by $f$ and $\hat{f}$. Since $\R$ is a BDE for $f$, for all $t > 0$, we have that
\begin{align*}
\hat{v}_{H_1}(t) & = v_{A_{00}}(t) &
\hat{v}_{H_2}(t) & = v_{A_{10}}(t) = v_{A_{01}}(t) \\
\hat{v}_{H_3}(t) & = v_{A_{11}}(t) &
\hat{v}_{H_4}(t) & = v_{B}(t) ,
\end{align*}
provided that the above identities are satisfied at $t = 0$.
\end{examplenum}

\paragraph*{Forward equivalence} We recall the definition of forward differential equivalence (FDE) from~\cite{popl16}.

\begin{definition}[Forward differential equivalence]\label{def:FDE}
Let $f$ be a vector field over $\X$. An equivalence relation $\R \subseteq \X \times \X$ is an FDE for $f$ if
$(x, y) \in \R$ implies
\begin{align*}
\sum_{z \in H} f_{z}(v) = \sum_{z \in H} f_{z}[x / \lambda (x + y) , y / (1 - \lambda) (x + y)](v)
\end{align*}
for all $H \in \X/\R$ and $v \in \RE^{\X \uplus \{\lambda\}}$, where $f[x / y]$ denotes the term arising when $x$ is replaced with $y$ in $f$.
%
\end{definition}

An FDE yields a self-consistent reduced ODE system that gives the dynamics of the sum of the original variables for each equivalence class~\cite[Theorem~1]{popl16}. Given an FDE $\R$ for the vector field $f$, one can define the corresponding $\R$-quotient vector field $\hat{f} \colon \RE^{\X/\R} \to \RE^{\X/\R}$ obtained from $f$ by a change of variable $H = \sum_{x \in H} x$, for each $H \in \X/\R$.
%


\begin{examplenum}\label{example:002fe}
Let $\R$ and $\X/\R = \{H_1,\ldots,H_4\}$ be as in Example~\ref{example:0022}.
The relation $\R$ can be also shown to be an FDE for $f$ in~\eqref{eq:running:example}.
The FDE-reduced vector field of $f$ w.r.t. $\R$ is
\begin{align*}
\hat{f}_{H_1} & = -4 {H_1} {H_4} + 3 {H_2} \\
\hat{f}_{H_2} & = 4 {H_1} {H_4} - 3 {H_2} - {H_2} {H_4} + 6 {H_3} \\
\hat{f}_{H_3} & = {H_2} {H_4} - 6 {H_3} \\
\hat{f}_{H_4} & = -4 {H_1} {H_4} + 3 {H_2} - {H_2} {H_4} + 6 {H_3}
\end{align*}
Let $v$ and $\hat{v}$ denote the solutions of the ODEs induced respectively by $f$ and $\hat{f}$.
Since $\R$ is an FDE for $f$, for all $t > 0$, we have that
\begin{align*}
\hat{v}_{H_1}(t) & = v_{A_{00}}(t) &
\hat{v}_{H_2}(t) & = v_{A_{10}}(t) + v_{A_{01}}(t) \\
\hat{v}_{H_3}(t) & = v_{A_{11}}(t) &
\hat{v}_{H_4}(t) & = v_{B}(t) ,
\end{align*}
provided that the above identities are satisfied at $t = 0$.
\end{examplenum}

In contrast to BDE, whose relation with the ODEs solution are conditional on the initial value $v(0)$, FDE does not make any such assumptions. At the same time, FDE preserves only the sums of the solutions of the original variables, while BDE preserves the solutions of the original variables in full. Although in the previous example the same relation is both an FDE and an BDE, the two notions are not comparable in general~\cite{concur15}. Indeed, this derives also from the fact that forward and backward bisimulations for Markov chains are not comparable.

For polynomial vector fields, the largest BDE and FDE exist and can be computed by a partition refinement algorithm with polynomial complexity in the size of the vector field~\cite{pnas17}. If applied to~(\ref{eq:running:example}), for instance, the algorithm confirms that $\R$ from Example~\ref{example:002} (reps.\ Example~\ref{example:002fe}) is the largest BDE (resp.\ FDE).

\section{Coupling method for polynomials} \label{sec:couplingmethod}

Our approach for reasoning about equivalences for polynomials is based on the \emph{coupling method} and on the celebrated proof technique by Strassen~\cite{Strassen} for checking stochastic dominance of random variables. Here we extend the concept of coupling to linear combinations and monomials and prove
Strassen-like theorems for these expressions. This will provide us with an
efficient method for checking equivalence of polynomials based on the existence
of suitable couplings.

\begin{definition}[Linear coupling] \label{def:linearcoupling}
Let $g,h$ be two linear combinations over $X$. A map $\omega \colon X \times X \to \RE_{\geq 0}$
is a \emph{linear coupling} for $(g,h)$ if the following conditions hold
\begin{enumerate}[label=(\roman*)]
 \item $\forall x \in X. \; \sum_{y \in X} \omega(x,y) = (g^{+} + h^{-})(x)$;
 \item  $\forall y \in X. \; \sum_{x \in X} \omega(x,y) = (h^{+} + g^{-})(y)$;
\end{enumerate}
where $g = g^{+} - g^{-}$ and $h = h^{+} - h^{-}$.

We denote by $\Gamma_\lin(g,h)$ the set of linear couplings for $(g,h)$.
Note that $\Gamma_\lin(g,h) \neq \emptyset$ iff $\sum_{x \in X} (g(x) - h(x)) = 0$.
\end{definition}


\begin{examplenum} \label{ex:linearcoupling}
Consider the linear combinations
\begin{align*}
g = \overbrace{2 x + 3 y}^{g^{+}}
- \big( \overbrace{3 z + z' }^{g^{-}} \big) \, , &&
h = \overbrace{2 x + 3 y}^{h^{+}}
- \big( \overbrace{3 w + w'}^{h^{-}} \big) \,.
\end{align*}
A linear coupling for $(g,h)$ is
\begin{align*}
\omega(y , y) = 2\,, &&
\omega(y , z) = 1\,, &&
\omega(w , y) = 1\,, \\
\omega(w , z) = 2\,, &&
\omega(x, x) = 2\,, &&
\omega(w', z') = 1\,,
\end{align*}
and $\omega(\cdot,\cdot) = 0$ otherwise.
Figure~\ref{fig:lin:coupling} provides a tabular visualization of the coupling $\omega$.
Another coupling for $(g,h)$ is
\begin{equation*}
\hat{\omega}(x, x) = 2\,, \quad
\hat{\omega}(y , y) = 3\,, \quad
\hat{\omega}(w', z') = 1\,, \quad
\hat{\omega}(w , z) = 3 \,,
\end{equation*}
and $\hat{\omega}(\cdot,\cdot) = 0$ otherwise.
\end{examplenum}

\begin{figure}[tp]
\begin{center}
\small
  \begin{tabular}{ c | c | c | c | c | c }
    $\omega$   & $x$ & $y$ & $z$ & $z'$  \\ \hline
    $x$ & 2 & 0 & 0 & 0 & $g^+(x)$ \\ \hline
    $y$ & 0 & 2 & 1 & 0 & $g^+(y)$ \\ \hline
    $w$ & 0 & 1 & 2 & 0 & $h^-(w)$ \\ \hline
    $w'$ & 0 & 0 & 0 & 1 & $h^-(w')$ \\ \hline
    & $h^{+}(x)$ & $h^{+}(y)$ & $g^{-}(z)$ & $g^{-}(z')$
  \end{tabular}
\end{center}
\caption{Tabular representation of the linear coupling $\omega$ for $(g,h)$ from Example~\ref{ex:linearcoupling}. As for Definition~\ref{def:linearcoupling}, the sum of each row
equals the corresponding coefficient in $g^+ + h^-$, while the sum of each column equals that of
$h^+ + g^-$.
}
\label{fig:lin:coupling}
\end{figure}

The following theorem establishes an equivalence between an existential property. \footnote{All proofs are available in the Appendix.}
\begin{theorem}\label{thm:linStrassen}
Let $\R \subseteq X \times X$ be an equivalence relation and $g,h$ two linear combinations over $X$.
Then, the following are equivalent:
\begin{enumerate}
\item \label{itm:linStrass1} There exists $\omega \in \Gamma_\lin(g,h)$ such that $\supp(\omega) \subseteq \R$.
\item  \label{itm:linStrass2} For all $v \in \RE^X$, $\bigwedge_{(x,y) \in \R} v_x \leq v_y$ implies $g(v) \leq h(v)$.
\item  \label{itm:linStrass3} For all $v \in \RE^X$, $\bigwedge_{(x,y) \in \R} v_x = v_y$ implies $g(v) = h(v)$.
\end{enumerate}
Moreover, $\eqref{itm:monStrass1} \Rightarrow \eqref{itm:monStrass2} \wedge \eqref{itm:monStrass3}$ holds for arbitrary relations $\R$.
\end{theorem}

\begin{remark}
The requirement of $\R$ being an equivalence cannot be relaxed as $g = 3x$, $h = x+y+z$ and
$\R = \{ (x,y), (y,z)\}$ is a counterexample.
\end{remark}

Next we provide an analogous  theorem for monomials. To this end we introduce the notion of monomial coupling.
\begin{definition}[Monomial coupling]\label{def:mon:coup}
Let $m,n$ be two monomials over $X$. A map $\rho \colon X \times X \to \RE_{\geq 0}$ is a \emph{monomial coupling} for $(m,n)$ if the following conditions hold:
\begin{enumerate}[label=(\roman*)]
\item $\forall x \in X. \; \sum_{y \in X} \rho(x,y) = m(x)$;
\item $\forall y \in X. \; \sum_{x \in X} \rho(x,y) = n(y)$.
\end{enumerate}
We denote by $\Gamma_\mon(m,n)$ the set of monomial couplings for $(m,n)$. Note that $\Gamma_\mon(m,n) \neq \emptyset$ iff $\sum_{x \in X} (m(x) - n(x)) = 0$.
\end{definition}


\begin{examplenum} \label{ex:moncoupling}
Consider the monomials ${A_{01}} {B}$ and ${A_{10}} {B}$ in~(\ref{eq:running:example}). A monomial coupling $\rho$ for pair the $({A_{01}} {B},{A_{10}} {B})$ is
\begin{align*}
\rho(A_{01},A_{10}) = 1, \ \
\rho(B,B) = 1 \ \ \text{and} \ \
\rho(\cdot,\cdot) = 0 \ \ \text{otherwise} .
\end{align*}
Again, there may be several couplings for the same pair, e.g.,
\begin{align*}
\hat{\rho}(A_{01},B) = 1, \ \
\hat{\rho}(B,A_{10}) = 1 \ \ \text{and} \ \
\hat{\rho}(\cdot,\cdot) = 0 \ \ \text{otherwise}
\end{align*}
is also a coupling for $({A_{01}} {B},{A_{10}} {B})$.
\end{examplenum}

The following is the variant of Strassen's theorem for monomials.
\begin{theorem}\label{thm:monStrassen}
Let $\R \subseteq X \times X$ be an equivalence relation and $m,n$ two monomials over $X$.
Then, the following are equivalent
\begin{enumerate}
\item  \label{itm:monStrass1} There exists $\rho \in \Gamma_\mon(m,n)$ such that $\supp(\rho) \subseteq \R$;
\item  \label{itm:monStrass2} For all $v \in \RE_{>0}^X$, $\bigwedge_{(x,y) \in \R} v_x \leq v_y$ implies $m(v) \leq n(v)$;
\item  \label{itm:monStrass3} For all $v \in \RE^X$, $\bigwedge_{(x,y) \in \R} v_x = v_y$ implies $m(v) = n(v)$.
\end{enumerate}
Moreover, $\eqref{itm:monStrass1} \Rightarrow \eqref{itm:monStrass2} \wedge \eqref{itm:monStrass3}$ holds for arbitrary relations $\R$.
\end{theorem}

\begin{remark}
The requirement of $\R$ being an equivalence cannot be relaxed as $m = x^3$, $n = xyz$, and
$\R = \{ (x,y), (y,z)\}$ is a counterexample.
\end{remark}

Theorems~\ref{thm:linStrassen} and~\ref{thm:monStrassen} provide us with a convenient method to prove the equivalence of polynomials by using the coupling method. Before stating this result, we introduce the following notation which will be useful in the reminder of the paper.

\begin{definition}[Liftings]
For a relation $\R \subseteq X \times X$, define
\begin{align*}
\lin[\R] &= \{ (g,h) \in  \lin \times \lin \mid \exists \omega \in \Gamma_\lin(g,h) \,. \supp(\omega) \subseteq \R \}
\, , \\
\mon[\R] &= \{ (m,n) \in \mon \times \mon \mid \exists \rho \in \Gamma_\mon(m,n) \,. \supp(\rho) \subseteq \R \}   \, .
\end{align*}
We call $\lin[\R]$ the lifting of $\R$ over linear combinations, and $\mon[\R]$ the lifting of $\R$ over monomials.
\end{definition}

Since we assume polynomials to always be expressed as linear combinations of monomials, we define the lifting of $\R$ over polynomials as $\poly[\R] = \lin[\mon[\R]]$. With this in place, we state our desired result.
\begin{corollary}\label{cor:couplingEquiv}
Let $\R \subseteq X \times X$ and $p,q$ two polynomials over $X$. Then, for all $v \in \RE^X$,
\begin{align*}
(p,q) \in \poly[\R]  && \text{implies} &&
\Big( \bigwedge_{(x,y) \in \R} v_x = v_y \Big) \Rightarrow p(v) = q(v) \, .
\end{align*}
If $\R$ is an equivalence, also the converse implication holds.
\end{corollary}


The following example illustrates how one can use Corollary~\ref{cor:couplingEquiv} to
prove equivalence among polynomials.
\begin{examplenum}
Consider the relation $\R$ from Example~\ref{example:002} and the
polynomials $f_{A_{01}}$ $f_{A_{10}}$ from~\eqref{eq:running:example}.
A linear coupling for $(f_{A_{01}},f_{A_{10}})$ is
\begin{align*}
\omega({A_{00}} B, {A_{00}} B) & = 2 &
\omega({A_{11}} , {A_{11}}) & = 3 \\
\omega({A_{10}} B, {A_{01}} B) & = 1 &
\omega({A_{10}} , {A_{01}}) & = 3
\end{align*}
and $\omega(\cdot,\cdot) = 0$ otherwise (\cf\ Example~\ref{ex:linearcoupling}). Note that the support of $\omega$ identifies a matching of monomials that, if assumed equal, imply $f_{A_{01}} = f_{A_{10}}$ (by Theorem~\ref{thm:linStrassen}) as illustrated below:
\begin{align*}
f_{A_{01}}  & = 2 \refnode{${A_{00}} B$}{m1} + 3 \refnode{${A_{11}}$}{m2} -
\big( 3 \refnode{${A_{01}}$}{m3} + \refnode{${A_{01}} B$}{m4} \big) \,, \\[2ex]
f_{A_{10}} & = 2 \refnode{${A_{00}} B$}{n1} + 3 \refnode{${A_{11}}$}{n2} -
\big( 3 \refnode{${A_{10}}$}{n3} + \refnode{${A_{10}} B$}{n4} \big) \,.
\end{align*}
\begin{tikzpicture}[remember picture,overlay]
\path[-,color=black]
	(m1) edge (n1)
	(m2) edge (n2)
	(n3) edge (m3)
	(n4) edge (m4)
	;
\end{tikzpicture}
Clearly $({A_{11}} , {A_{11}}), ({A_{10}} , {A_{01}}) \in \mon[\R]$, and as illustrated in
Example~\ref{ex:moncoupling}, also $({A_{00}} B, {A_{00}} B), ({A_{10}} B, {A_{01}} B) \in \mon[\R]$. Therefore, $(f_{A_{01}}, f_{A_{10}}) \in \poly[\R]$.
Again, the supports of the monomial couplings identify a matching of the variables that, by Corollary~\ref{cor:couplingEquiv}, when assumed to be equal, imply $f_{A_{01}} = f_{A_{10}}$ as illustrated below:
\begin{align*}
f_{A_{01}}  & = 2 \refnode{${A_{00}}$}{M11} \, \refnode{$B$}{M12}
- 3\refnode{${A_{01}}$}{M2}
- \refnode{${A_{01}}$}{M31} \, \refnode{$B$}{M32}
+ 3 \refnode{${A_{11}}$}{M4} \,, \\[2ex]
f_{A_{10}} & = 2 \refnode{${A_{00}}$}{N11} \, \refnode{$B$}{N12}
- 3\refnode{${A_{10}}$}{N2}
- \refnode{${A_{10}}$}{N31} \, \refnode{$B$}{N32}
+ 3 \refnode{${A_{11}}$}{N4} \,.
\end{align*}
\begin{tikzpicture}[remember picture,overlay]
\path[-,color=black]
	(M11) edge (N11)
	(M12) edge (N12)
	(M2) edge (N2)
	(M31) edge (N31)
	(M32) edge (N32)
	(M4) edge (N4)
	;
\end{tikzpicture}
Note that this is in line with the fact that $\R$ is a BDE for the polynomial vector field $f$ in~\eqref{eq:running:example}.
\end{examplenum}



\section{From Couplings to Differential Equivalences} \label{sec:BDBandFBD}

Here we give a coinductive characterization of BDE and FDE over polynomial vector fields as the greatest fixed point of two monotone operators whose definition is based on the notion of couplings. This will give us a coinduction proof principle which exploits coupling-based methods and will constitute the formal basis of our \ch{local} algorithm (Section~\ref{sec:onthefly}).

\newcommand{\Rel}[1][S]{2^{#1 \times #1}}
\paragraph*{Coinduction proof principle} For a monotone map $b$ on the lattice of relations, the Knaster-Tarski fixed-point theorem characterizes the greatest fixed point $\gfp{b}$ as the greatest lower bound of all its post-fixed points (i.e., $\gfp{b} = \bigcup \{\R \mid \R \subseteq b(\R)\}$). This leads to the coinduction proof principle illustrated below:
\begin{equation*}
 \frac{\exists \R'\,, \R \subseteq \R' \subseteq b(\R') }{ \R \subseteq \gfp{b}} \,.
\end{equation*}
By characterising BDE (resp.\ FDE) as the greatest fixed point of some monotone operator, we can exploit the above principle to show that a set $\R$ is contained in a BDE (resp.\ FDE) by providing a post-fixed point $\R'$ containing it. The operators will be defined in terms of a \emph{backward} (resp., \emph{forward}) \emph{differential bisimulation}, discussed next.

\subsection{Backward Differential Bisimulation}


\begin{definition} \label{def:BDB}
Let $f$ be a polynomial vector field over $\X$. A relation $\R \subseteq \X \times \X$
is a \emph{backward differential bisimulation} (BDB) for $f$, if it is a post-fixed point of
the following operator:
\begin{align*}
\mathcal{B}^f(\R) = \{ (x,y) \in \X \times \X \mid (f_x, f_y) \in \poly[\R] \} .
\end{align*}
\end{definition}
It is easy to show that $\mathcal{B}^f$ is monotone in the lattice of relations, therefore $\gfp{\mathcal{B}^f}$ exists and is the greatest BDB.
In the remainder, whenever the vector field $f$ is clear from the context, we write $\mathcal{B}$ in place of $\mathcal{B}^f$.

Thanks to Corollary~\ref{cor:couplingEquiv}, we prove the following result. It ensures that finding a BDB, instead of a BDE, is enough to imply that related variables have the same solutions if initialized equally (cf. \cite[Theorem~3]{popl16}).
\begin{proposition} \label{prop:BR}
Let $f \colon \RE^\X \to \RE^\X$ be polynomial vector field and $\R$ a BDB. Then, for any $v \in \RE^\X$ we have
\begin{equation*}
\Big( \bigwedge_{(x,y) \in \R} v_x = v_y \Big) \Rightarrow \Big( \bigwedge_{(x,y) \in \R} f_x(v) = f_y(v) \Big) \, .
\end{equation*}
\end{proposition}

\begin{examplenum}\label{example:003}
Consider the vector field $f : \RE^\X \to \RE^\X$ from~(\ref{eq:running:example}) and let
$\R' = \mathit{id} \cup \{(A_{01},A_{10})\}$. Then, it can be shown that
\begin{equation*}
 \{ (3 {A_{01}} , 3 {A_{10}}) , ({A_{01}} B, {A_{10}} B) \} \subseteq \mon[\R']
\end{equation*}
From this, it is possible to infer that $ (f_{A_{10}},f_{A_{01}})  \in \poly[\R']$. It can be noted that $\R'$ is a BDB but not a BDE because it is not an equivalence relation.
\end{examplenum}
The above example shows that not every BDB is a BDE. The following theorem clarifies the close connection among the notions of BDB and BDE.

\begin{theorem}[Fixed-point characterization of BDE]\label{thm:BR2BE}
Let $f$ be a polynomial vector field over $\X$ and $\R \subseteq \X \times \X$.
Then, the following hold:
\begin{enumerate}[label=(\roman*), leftmargin=5ex]
    \item \label{itm:BR2BE1} If $\R$ is a BDB then $e(\R)$ is a BDE.
    \item \label{itm:BR2BE2} If $\R$ is a BDE then $\R$ is a BDB.
    \item \label{itm:BR2BE3} $\gfp{\mathcal{B}}$ is the greatest BDE.
\end{enumerate}
\end{theorem}


\subsection{Forward Differential Bisimulation}

Forward differential bisimulation (FDB) can be tied to the notion of \emph{total derivative of a vector field} $f \colon \RE^\X \to \RE^\X$, commonly denoted by $\partial f$. It is given by the Jacobian matrix $(\partial_{x_j} f_{x_i})_{x_i,x_j \in \X}$, where $\partial_{x_j} f_{x_i}$ denotes the partial derivative of $f_{x_i} \colon \RE^\X \to \RE$ with respect to $x_j$. If $f$ is a polynomial vector field, it is well-known (e.g.,~\cite{Li19891413}) that the total derivative can be written as
\begin{align}\label{eq:jac:dec}
(\partial_{x_j} f_{x_i})_{x_i,x_j \in \X} = \textstyle\sum_{k = 1}^\kappa m_k \cdot J_k ,
\end{align}
where $m_1,\ldots,m_\kappa$ are pairwise different monomials over $\X$ and $J_1,\ldots,J_\kappa \in \RE^{\X \times \X}$.

\begin{examplenum}\label{example:jac:dec}
In the case of our running example~(\ref{eq:running:example}), the Jacobian matrix $(\partial_{x_j} f_{x_i})_{x_i,x_j \in \X}$ is given by
\[
\left(
    \begin{array}{rrrrr}
    -4{B} &                    3  &          3        &  0   &  -4{A_{00}} \\
    2{B}  &     -{B} \! - \! 3  &          0        &  3   &  2{A_{00}} \! - \! {A_{01}} \\
    2{B}  &                    0  &  -{B} \! - \! 3 &  3   &  2{A_{00}} \! - \! {A_{10}} \\
    0       &                {B}  &       {B}       & -6   &  {A_{01}} \! + \! {A_{10}} \\
    -4{B} &      3 \! - \! {B}  & 3 \! - \! {B}   &  6   & -4{A_{00}} \! - \! {A_{01}} \! - \! {A_{10}} \\
    \end{array}
\right)
\]
The decomposition~(\ref{eq:jac:dec}) can be expressed by five matrices $J_1$, $J_{B}$, $J_{A_{00}}$, $J_{A_{01}}$ and $J_{A_{10}}$, where $J_1$ accounts for the constant monomial, while $J_{x_i}$ accounts for the monomial ${x_i}$. For instance, $m_1 = 1$ and $m_B = B$, while
\[
\underbrace{\left(
    \begin{array}{rrrrr}
    0  &  3  &  3  &  0   &  0 \\
    0  & -3  &  0  &  3   &  0 \\
    0  &  0  & -3  &  3   &  0 \\
    0  &  0  &  0  & -6   &  0 \\
    0  &  3  &  3  &  6   &  0 \\
    \end{array}
\right)}_{ J_1 }
\quad
\underbrace{\left(
    \begin{array}{rrrrr}
    -4 &      0  &     0  &  0   &  0 \\
    2  &     -1  &     0  &  0   &  0 \\
    2  &      0  &    -1  &  0   &  0 \\
    0  &      1  &     1  &  0   &  0 \\
    -4 &      -1 &    -1  &  0   &  0 \\
    \end{array}
\right)}_{ J_B }
\]
\end{examplenum}

We are now in the position to introduce the operator that will be used for the fixed-point characterization of FDE.
\begin{definition} \label{def:FBD}
Let $f$ be a polynomial vector field over $\X$ and $F = \{J_1,\ldots,J_\kappa\}$ be as in~(\ref{eq:jac:dec}). A relation $\R \subseteq \X \times \X$ is a \emph{forward differential bisimulation} (FDB) for $f$ if it is a post-fixed point of the following operator:
\begin{align*}
\calF^f(\R) = \bigcap_{J \in F} \mathcal{B}^{J^T}(\R) \,.
\end{align*}
\end{definition}
Essentially, $\R$ is an FDB for $f$ whenever it is a BDB of all $J_1^T, \ldots, J_\kappa^T$ in~(\ref{eq:jac:dec}), where $A^T$ is the transpose of matrix $A$.

The monotonicity of $\calF^f$ follows by that of $\mathcal{B}$. Therefore, by Knaster-Tarski fixed-point theorem, $\gfp{\calF^f}$ exists. In the following, whenever the vector field $f$ is clear from the context, we write $\calF$ in place of $\calF^f$.


The following pivotal observation relies on~\cite[Lemma I.1]{DBLP:journals/corr/abs-2004-11961} and reduces FDE to BDE.

\begin{theorem}\label{thm:fe:via:be}
Fix some index set $\X$, some equivalence relation $\R \subseteq \X \times \X$ and a polynomial vector field $f : \RE^\X \to \RE^\X$. Then, $\R$ is an FDE of $f$ if and only if $\R$ is a BDE of each linear vector field $J_k^T : \RE^\X \to \RE^\X, v \mapsto J_k^T v$ from~(\ref{eq:jac:dec}).
\end{theorem}

We can now establish the connection between FDB and FDE, similarly to Theorem~\ref{thm:BR2BE}.
\begin{theorem}(Fixed-point characterization of FDE)\label{thm:fe:on:the:fly}
Let $f$ be a polynomial vector field over $\X$ and $\R \subseteq \X \times \X$.
Then, the following hold:
\begin{enumerate}[label=(\roman*), leftmargin=5ex]
\item \label{itm:thm:FDB1} If $\R$ is an FDB, then $e(\R)$ is an FDE.
\item \label{itm:thm:FDB2} If $\R$ is an FDE, then $\R$ is an FDB.
\item  \label{itm:thm:FDB3} The $\gfp{\mathcal{F}}$ is the greatest FDB.
\end{enumerate}
\end{theorem}

\begin{examplenum}
It can be noted that $\R$ from Example~\ref{example:002} is an FDB because $\R$ is a BDB of $J_1^T$, $J_{B}^T$, $J_{A_{00}}^T$, $J_{A_{01}}^T$ and $J_{A_{11}}^T$.
\end{examplenum}

\subsection{Constrained bisimulation}\label{sec:constrainedAlg}
In line with Proposition~\ref{prop:BR}, a BDB imposes pre-conditions on the initialization of related variables. As anticipated in Section~\ref{sec:intro}, from a modeling viewpoint one may want to express the need of \emph{not} relating variables. A constrained BDB excludes certain pairs from being used in the relation.
\begin{definition} \label{def:constrBDB}
Let $f$ be a polynomial vector field over $\X$ and $C \subseteq \X \times \X$.
A relation $\R \subseteq \X \times \X$ is a \emph{$C$-constrained BDB} for $f$,
if $\R$ is a BDB and $\R \cap C = \emptyset$.
\end{definition}

Next, we show that also the the concepts of $C$-constrained BDB has a coindunction proof principle. To this end, consider the following operator:
\begin{equation}
\mathcal{B}^f_{C}(\R) = \mathcal{B}^f(\R \setminus C) \setminus C \label{eq:BC}
\end{equation}
As usual, when the vector field $f$ is clear from the context, we write $\mathcal{B}_C$ in place of $\mathcal{B}^f_C$; and when $C = \emptyset$, we write $\mathcal{B}$ in place of $\mathcal{B}_C$.

\begin{theorem}\label{thm:constrBR}
Let $f$ be a polynomial vector field over $\X$ and $\R, C \subseteq \X \times \X$.
Then, $\R$ is a $C$-constrained BDB for $f$ iff $\R \subseteq \mathcal{B}_{C}(\R)$.
\end{theorem}

For arbitrary $C \subseteq \X \times \X$, $\mathcal{B}^f_C$ is monotone because $\mathcal{B}^f$ is. Hence, by Knaster-Tarski fixed-point theorem, $\gfp{\mathcal{B}^f_C}$ exists and, by Theorem~\ref{thm:constrBR}, is the greatest $C$-constrained BDB.

\begin{remark}
By Theorem~\ref{thm:BR2BE} and~\ref{thm:constrBR} we have that if $\R$ is a $C$-constrained BDB, then $e(\R)$ is a BDE.
Note however that, depending on the choice of $C \subseteq \X \times \X$, it may be the case that $e(\R) \cap C \neq \emptyset$. In some cases, one needs to impose additional constraints to be able to find a $C$-constrained BDB which extends to a $C$-constrained BDE.
Nevertheless, proving that no $C$-constrained BDB exists, directly implies that there is no BDE $\R$ with $\R \cap C = \emptyset$.
\end{remark} 


\section{Local Algorithm}\label{sec:onthefly}


As described in Section~\ref{sec:BDBandFBD}, one can exploit the coinduction proof principle to show that two variables $x$ and $y$ are related by a BDE, by providing a BDB $\R$ that contains $(x,y)$.
The classical approach based on partition refinement consists in computing $\mathtt{gfp}(\mathcal{B})$ as the limit of the decreasing chain $\{ \mathcal{B}^{i}(\X \times \X) \}_{i \in \mathbb{N}}$.
Here we propose a \ch{local} approach that, starting from a relation $\R \subseteq \X \times \X$ containing some query pairs, iteratively updates it until $\R$ is proven to be a BDB.
\ch{This is done by performing a local exploration of the system, led by the dependencies that are discovered at each iteration.} The following example illustrates the intuition behind our algorithm.
\begin{examplenum} \label{ex:otheflylucky}
Consider the vector field in Eq.~\eqref{eq:running:example} and suppose we want to prove $A_{01}$ and $A_{10}$ to be backward equivalent. 
We do so by finding a relation $\R$ such that $(A_{01}, A_{10}) \in \R$ and $\R \subseteq \mathcal{B}(\R)$.
We start with $\R = \{(A_{01}, A_{10}) \}$, but we notice that $(f_{A_{01}}, f_{A_{10}}) \not\in \poly[\R]$.
In line with Corollary~\ref{cor:couplingEquiv}, to equate $f_{A_{01}}$ and $f_{A_{10}}$ we may add $(A_{00}, A_{00})$, $(B,B)$, and $(A_{11},A_{11})$ to $\R$ as illustrated below:
\begin{align*}
f_{A_{01}} &= 2 {\refnode{$A_{00}$}{A00-1}} \, {\refnode{$B$}{B-1}}
	- 3 {\refnode{$A_{01}$}{A01-1}} - {\refnode{$A_{01}$}{A01-1bis}} \, {\refnode{$B$}{B-1bis}} + 3 {\refnode{$A_{11}$}{A11-1}} \\[2ex]
f_{A_{10}} &= 2 {\refnode{$A_{00}$}{A00-2}} \, {\refnode{$B$}{B-2}}
	- 3 {\refnode{$A_{10}$}{A10-2}} - {\refnode{$A_{10}$}{A10-2bis}} \, {\refnode{$B$}{B-2bis}} + 3 {{\refnode{$A_{11}$}{A11-2}}}
\end{align*}
\begin{tikzpicture}[remember picture,overlay]
\path[-]
	(A00-1) edge (A00-2)
	(A11-1) edge (A11-2)
	(B-1) edge (B-2)
	(B-1bis) edge (B-2bis)
	;
\path[-,color=gray]
	(A01-1) edge (A10-2)
	(A01-1bis) edge (A10-2bis)
	;
\end{tikzpicture}
While this yields $(f_{A_{01}}, f_{A_{10}}) \in \poly[\R]$, for all new pairs in $\R$ we need to check that the corresponding pairs of polynomials are in $\poly[\R]$. This is not the case yet, since $(A_{00}, A_{00}) \not\in \poly[\R]$. To remedy this, we add $(A_{10}, A_{01})$ to $\R$:
\begin{align*}
f_{A_{00}} & = -4 {\refnode{$A_{00}$}{A00-1}} \, {\refnode{$B$}{B-1}} + 3 {\refnode{$A_{10}$}{A10-1}} + 3 {\refnode{$A_{01}$}{A01-1}} \\[2ex]
f_{A_{00}} & = -4 {\refnode{$A_{00}$}{A00-2}} \, {\refnode{$B$}{B-2}} + 3 {\refnode{$A_{10}$}{A10-2}} + 3 {\refnode{$A_{01}$}{A01-2}}
\end{align*}
\begin{tikzpicture}[remember picture,overlay]
\path[-]
	(A10-1) edge[bend left=10] (A01-2)
	;
\path[-,color=gray]
	(A00-1) edge (A00-2)
	(B-1) edge (B-2)
	(A01-1) edge[bend right=10] (A10-2)
	;
\end{tikzpicture}
This yields the following relation
\begin{equation*}
\R = \{(A_{01}, A_{10}), (A_{10}, A_{01}), (A_{00}, A_{00}), (A_{11},A_{11}), (B,B)\} \,,
\end{equation*}
which one can verify to be a BDB. To equate $f_{A_{00}}$ with itself, in the last step, we could have added the pairs $(A_{10},A_{10})$ and $(A_{01},A_{01})$ instead. However, the above choice resulted in a smaller BDB.
\end{examplenum}

Algorithm~\ref{alg:otf} implements a procedure that, given a query set $\id{Query} \subseteq \X \times \X$, discovers which query pairs can be related by some $C$-constrained BDB. This is carried out by constructing a $C$-constrained BDB $\R$ \ch{following a local exploration of the system} as illustrated in Example~\ref{ex:otheflylucky}.

Starting from $\R = \id{Query} \setminus C$,  and $\hat{\R} = C$ the algorithm iteratively updates $\R$ and $\hat{\R}$ maintaining the invariant $\R \cap \hat{\R} = \emptyset$ until $\R$ is a $C$-constrained BDB. At each iteration of the while-loop (lines~\ref{alg:whilebegin}--\ref{alg:whileend}) we check if, whenever $(x,y) \in \R$, $(f_x, f_y) \in \poly[\R]$ by searching for some couplings having support disjoint from $\hat{\R}$. If this is possible, then we add the support of such couplings to $\R$; otherwise $x$ and $y$ cannot be related by any $C$-constrained BDB and we move $(x,y)$ from $\R$ to $\hat{\R}$. In line with $\poly[\R] = \lin[\mon[\R]]$ we simplify the search of linear couplings and monomial couplings by using $\Q$ and $\hat{\Q}$. The former mimics $\mon[\R]$, while the latter stores monomial pairs which have been found not possible to relate as the lifting of some $C$-constrained BDB.
The while-loop terminates when $\R$ and $\Q$ \emph{have not} changed with respect to the previous iteration.
%
%
%
\begin{algorithm}[t]
\vspace{-0.3cm}
\begin{codebox}
\Procname{$\onthefly(f \colon \RE^\X \to \RE^\X, \id{Query} \subseteq \X \times \X, \id{C} \subseteq \X \times \X)$}
\li initialize $\R \gets \id{Query} \setminus C$ and $\hat{\R} \gets C$  \label{alg:init}
\li initialize $\Q \gets \emptyset$ and $\hat{\Q} \gets \emptyset$
\li $\id{Prev} = (\emptyset, \emptyset,\emptyset, \emptyset)$ \label{alg:initend}
\li \While $\id{Prev} \neq (\R,\hat{\R},\Q,\hat{\Q})$ \Do \label{alg:whilebegin}
\li $\id{Prev} = (\R,\hat{\R},\Q,\hat{\Q})$
\li \For \textbf{each} $(x,y) \in \R$ \Do \label{alg:loopRbegin}
\li \If $\exists \omega \in \Gamma_\lin(f_x,f_y)$ s.t.\ $\supp(\omega) \cap \hat{\Q} = \emptyset$ \Then \label{alg:ifOmega}
\li $\Q \gets \Q \cup \supp(\omega)$ \label{alg:up-to-Mg}
\li \Else
\li move $(x,y)$ from $\R$ to $\hat{\R}$ \label{alg:moveR}
\End
\End \label{alg:loopRend}
\li \For \textbf{each} $(m,n) \in \Q$ \Do \label{alg:loopQbegin}
\li \If $\exists \rho \in \Gamma_\mon(m,n)$ s.t.\ $\supp(\rho) \cap \hat{\R} = \emptyset$ \Then \label{alg:ifRho}
\li $\R \gets \R \cup \supp(\rho)$ \label{alg:up-to-g}
\li \Else
\li move $(m,n)$ from $\Q$ to $\hat{\Q}$ \label{alg:movetoQhat}
\End
\End  \label{alg:loopQend}
\End \label{alg:whileend}
\li \Return $\R$
\end{codebox}
\vspace{-0.3cm}
\caption{Local computation of BDB}\label{alg:otf}
\end{algorithm}

The following result states the correctness of Algorithm~\ref{alg:otf}.
\begin{theorem}\label{thm:onthefy}
Let $f$ be a polynomial vector field over $\X$ and $\id{Query}, C \subseteq \X \times \X$.
Then $\onthefly(f,\id{Query},C)$ is terminating and returns a relation $\R \subseteq \X \times \X$ such that:
\begin{enumerate}[label=(\roman*)]
\item $\R$ is a $C$-constrained BDB; 
\item if $(x,y) \in \id{Query}$, then $(x, y) \in \R$ iff $(x,y) \in \mathtt{gfp}(\mathcal{B}_{\id{C}})$
\end{enumerate}
\end{theorem}

Lines~\ref{alg:ifOmega} and~\ref{alg:ifRho} in Algorithm~\ref{alg:otf} can be respectively implemented by solving a linear program. Specifically, computing a monomial coupling $\omega$ satisfying the condition in Line~\ref{alg:ifOmega} can be done by solving:
\begin{equation*}
\begin{aligned}[c]
v = \min_{\omega} & \:
\textstyle \sum_{(m,n) \in \hat{\Q}} \omega(m,n) \\[-0.5ex]
	&\textstyle \sum_{n} \omega(m,n) = (f_x^{+} + f_y^{-})(m) && \forall m \in \mon \\
	&\textstyle \sum_{m} \omega(m,n) = (f_y^{+} + f_x^{-})(n) && \forall n \in \mon  \\
	& \omega(m,n) \geq 0 && \forall m,n \in \mon
\end{aligned}
\end{equation*}
Each feasible solution corresponds to a linear coupling $\omega \in \Gamma_\lin(f_x,f_y)$. In particular, the optimal value $v$ is $0$ if and only if $\omega(m,n) > 0$ implies $(m,n) \not\in \hat{\Q}$, i.e., $\supp(\omega) \cap \hat{\Q} = \emptyset$.

Analogously, one can find a monomial coupling satisfying the condition in Line~\ref{alg:ifRho} by solving the following linear program
\begin{equation*}
v = \min_{\rho \in \Gamma_\mon(m,n)} \textstyle \sum_{(x,y) \in \hat{\R}} \rho(x,y)
\end{equation*}
As before, monomial couplings $\rho \in \Gamma_\mon(m,n)$ are modeled by means of linear constraints, and the optimal value $v$ equals zero if and only if $\supp(\rho) \cap \hat{\R} = \emptyset$.

\begin{remark}
The above linear programs are both instances of the transportation problem~\cite{Dantzig51,FordF56}, which in turn is a well-known instance of the (uncapacitated) min-cost network flow problem. Notably, it can be solved efficiently, e.g., by employing Orlin's algorithm~\cite{Olrin88}.
\end{remark}

This ensures that Algorithm~\ref{alg:otf} runs in polynomial time.

\begin{theorem}\label{thm:alg:otf:comp}
Algorithm~\ref{alg:otf} runs in time $\mathcal{O}(  |\mon_f|^4 (k^2 + h^2))$ where $\mon_f$ is the set of monomials occurring in $f$, $k$ is the maximum number of monomials occurring in a single polynomial expression $f_x$ in $f$, and $h$ is the maximum number of variables occurring in a monomial expression in $\mon_f$.
\end{theorem}

Thanks to Theorems~\ref{thm:fe:via:be} and \ref{thm:fe:on:the:fly}, we can provide a \ch{local} procedure to compute FDB in a similar fashion as Algorithm~\ref{alg:otf}. For this, it is worth noting that the decomposition of the Jacobian matrix as in Equation~\eqref{eq:jac:dec} gives rise to a set $\{J_1,\twodots,J_\kappa\}$ of \emph{linear} vector fields over $\X$.  Then, we can approach the construction of an FDB in a similar fashion as for Algorithm~\ref{alg:otf}. This time, however, in line with the definition of the operator $\cal{F}$ one has to check, at each step, for the existence of a number linear couplings arising from the set of linear vector fields $\{J_1^T,\twodots,J_\kappa^T\}$. This is illustrated in Algorithm~\ref{alg:findFBD}.
\newcommand{\findFDB}{\proc{FindFDB}}
\begin{algorithm}[tp]
\vspace{-0.3cm}
\begin{codebox}
\Procname{$\findFDB(f \colon \RE^\X \to \RE^\X, \id{Query} \subseteq \X \times \X)$} 

\li initialize $\R \gets \id{Query}$ and $\hat{\R} \gets \emptyset$ \label{lin:FDBinit}
\li compute $\{J_1,\twodots,J_\kappa\}$ as in~\eqref{eq:jac:dec} and let $F = \{J^T_1,\twodots,J^T_\kappa\}$ \label{lin:jacdec}
\li $\id{Prev} = (\emptyset,\emptyset)$
\li \While $Prev \neq (\R,\hat{\R})$ \Do \label{lin:whilefdb}
\li $\id{Prev} = (\R,\hat{\R})$
\li \For \textbf{each} $(x,y) \in \R$ \Do \label{lin:kforbegin}
\li \If $\forall f \mathbin{\in} F .\,\exists \omega_f \mathbin{\in} \Gamma_\lin(f_x,f_y)$. $\supp(\omega_f) \cap \hat{\R} \mathbin{=} \emptyset$  \Then \label{alg:omegafk}
\li $\R \gets \R \cup \bigcup_{f \in F} \supp(\omega_f)$ \label{lin:kR}
\li \Else
\li move $(x,y)$ from $\R$ to $\hat{\R}$ \label{lin:kmove}
\End
\End \label{lin:kforend}
\End \label{lin:kwhileend}
\li \Return $\R$
\end{codebox}
\vspace{-0.3cm}
\caption{Local computation of FDB }\label{alg:findFBD} 
\end{algorithm}
The following result states its correctness.
\begin{theorem}  \label{thm:findFDB}
Let $f$ be a polynomial vector field over $\X$ and $\id{Query} \subseteq \X \times \X$.
Then $\findFDB(f,\id{Query})$ is terminating and returns a relation $\R \subseteq \X \times \X$ such that
\begin{enumerate}[label=(\roman*)]
\item $\R$ is a FDB;
\item if $(x,y) \in \id{Query}$, then $(x,y) \in \R$ iff $(x,y) \in \mathtt{gfp}(\cal{F})$.
\end{enumerate}
\end{theorem}


\begin{remark}
One may wonder whether it is possible to devise a direct \ch{local} algorithm for FDE. Unfortunately, while allowing for a global approach in form of partition refinement, FDE appears to escape a direct \ch{local} construction. More specifically, recall that in the case of BDE the \ch{local} approach was informed by the fact a BDE pair $(x,y)$ implied the existence of a coupling for $(f_x,f_y)$ which, in turn, ensured the existence of a family of monomial couplings, thus giving rise to a family of further BDE pairs and so on. In contrast, an FDE pair $(x,y)$ implies the existence of partitioning of the variables $H_1,\ldots,H_\nu \subseteq \X$ satisfying the condition of Definition~\ref{def:FDE}. Since the partitioning corresponds exactly to the sought FDE, it is not clear how a \ch{local} approach for FDE may proceed beyond this point.
\end{remark}

We conclude the subsection by discussing the computational complexity of Algorithm~\ref{alg:findFBD}.
The following lemma states that the decomposition~(\ref{eq:jac:dec}) can be computed in polynomial time.
\begin{lemma}\label{lem:jac:dec}
Let $f$ be a polynomial vector field over $\X$. Then, the $J_1,\ldots,J_\kappa$ from~(\ref{eq:jac:dec}) can be computed in $\mathcal{O}(|\X| |\mon_f|)$ steps, where $\mon_f$ is the set of all monomials occurring $f$. Moreover, it holds that $\kappa \leq |\X| |\mon_f|$ and $\sum_{k = 1}^\kappa |J_k| \leq |\X| |\mon_f|$, where $|A|$ is the number of non-zero entries in a matrix $A$.
\end{lemma}


Now we can provide the complexity of Algorithm~\ref{alg:findFBD}.
\begin{theorem} \label{thm:complexityFindFDB}
Let $f$ be a polynomial vector field over $\X$. Algorithm~\ref{alg:findFBD} runs in time $\mathcal{O}(|\X|^7 |\mon_f|)$ where $\mon_f$ is the set of monomials occurring in $f$.
\end{theorem}


%

\subsection{Employing up-to techniques}\label{sec:upto}
In this section we briefly recall up-to techniques~\cite{PousS12}, and describe how to enhance the coinduction proof principle which is at the heart of Algorithms~\ref{alg:otf} and~\ref{alg:findFBD}.

For a simpler treatment of up-to techniques that is generic on the monotone operator $b$ defining the different notions of bisimulation encountered so far, we will call $b$-simulation a post-fixed point of $b$. In many situations the computation of a $b$-simulation can be significantly optimized, if instead of computing a post-fixed point of $b$ one exhibits a relaxed invariant, that is, a relation $\R$ such that $\R \subseteq b(g(\R))$ for a suitable monotone function $g$.
The function $g$ is called a \emph{sound up-to technique} when the following proof principle is valid
\begin{equation*}
  \frac{(x,y) \in \R \quad \R \subseteq b(g(\R))}{(x,y) \in \gfp{b}} \,.
\end{equation*}
As clear from the above discussion, up-to techniques fit to use when one wants to check if $x$ and $y$ can be related by some $b$-simulation leaving implicit the construction of an actual $b$-simulation, which is replaced instead by a $(b \circ g)$-simulation.

The above notions are formalized in the following definition.
\begin{definition}[Simulation up-to, soundness] Let $b,g$ be two monotone functions.
A $b$-simulation up-to $g$ is a $(b \circ g)$-simulation. The function $g$ is $b$-sound if $\gfp{b \circ g} \subseteq \gfp{b}$. 
\end{definition}

Establishing the soundness of up-to techniques on a case-by-case basis can be tedious.
For this reason~\cite{PousS12} describes a framework giving sufficient conditions for proving soundness in a modular fashion, based on the notion of \emph{compatible} function.
\begin{definition}
Let $b,g$ be two monotone functions. The function $g$ is $b$-compatible if $g \circ b \subseteq b \circ g$.
\end{definition}
Compatible functions are also sound up-to techniques~\cite[Theorem 6.3.9]{PousS12} and, most importantly, they can be composed in several ways yielding another compatible function~\cite[Proposition 6.3.11]{PousS12}. Therefore, whenever possible, we will prove compatibility in place of soundness.

Our interest in up-to techniques has to do with the construction of a smaller witness proving which query pairs are contained in some BDB (resp.\ FDB). This will be implemented by computing a $\cal{B}_C$-simulation up-to $g$ for some $\cal{B}_C$-sound function $g$ given as input to the algorithm. For this, we additionally require $g$ to be an \emph{extensive} function, that is $\R \subseteq g(\R)$ for any relation $\R$.
Extensive functions will be conveniently used to reduce the size of the $b$-simulations 
up-to to be constructed in each step of the algorithm.

Given an extensive and $\cal{B}_C$-sound up-to technique $g$, we modify  \onthefly\ replacing Lines~\ref{alg:up-to-Mg} and~\ref{alg:up-to-g}  with
\begin{align}
R & \gets R \cup (\supp(\rho) \setminus g(R)) \ ,  \label{eq:updateR} \\
Q & \gets Q \cup (\supp(\omega) \setminus \mon[g(R) \setminus C]) \ ,  \label{eq:updateQ}
\end{align}
respectively. Note that, thanks to the fact that $g$ is extensive, instead of adding to $R$ (resp.\ $Q$) the entire support of the coupling $\rho$ (resp. $\omega$), we add only those pairs which cannot already be safely inferred from the current value of $R$ by exploiting the up-to proof principle. This makes the algorithm more conservative in expanding the size of $R$ (resp.\ $Q$), thus leading to a significant speed-up in performance as demonstrated in the next section.


The following result states that the above explained modification of Algorithm~\ref{alg:otf} returns a $\cal{B}_C$-simulations up-to $g$ witnessing which query pairs are contained in some $C$-constrained BDB.
\begin{theorem}\label{thm:uptoAlg}
Let $f$ be a polynomial vector field over $\X$, $g$ an \emph{extensive} and \emph{monotone} $\cal{B}_C$-sound up-to technique, and $\id{Query}, C \subseteq \X \times \X$. Then, $\onthefly(f, \id{Query}, C, g)$ is terminating and returns a relation $\R \subseteq \X \times \X$ such that:
\begin{enumerate}[label=(\roman*)]
\item \label{itm:1:thm:uptoAlg} $\R$ is a $\cal{B}_C$-simulation up-to $g$;
\item \label{itm:2:thm:uptoAlg} if $(x,y) \in \id{Query}$, then $(x, y) \in \R$ iff $(x, y) \in \mathtt{gfp}(\mathcal{B}_C)$.
\end{enumerate}
\end{theorem}


The following lemma states that reflexive, symmetric, transitive, and equivalence closures are $\cal{B}$-compatible up-to techniques, which can be safely used as input in Algorithm~\ref{alg:otf}.

\begin{lemma}\label{lem:B-compatible}
$r$, $s$, $t$, and $e$ are $\mathcal{B}$-compatible up-to techniques.
\end{lemma}
\begin{remark}
\ch{Notably, the proof of Lemma~\ref{lem:B-compatible} relies on the fact that $g \circ \mathcal{B} \circ g = \mathcal{B} \circ g$ for all $g \in \{r,s,t,e\}$. This also implies that for any $g \in \{r,s,t,e\}$ and $\R \subseteq \X \times \X$ $\mathcal{B}$-simulation up-to $g$, one can retrieve the corresponding $\mathcal{B}$-simulation as $g(\R)$. Indeed, $g(\R) \subseteq g(\mathcal{B}(g(\R))) = \mathcal{B}(g(\R))$
where the first inclusion holds true because $g$ is extensive.}
\end{remark}

\begin{examplenum}
Consider the polynomial vector field $f$ of Eq.~\eqref{eq:running:example}. The call $\onthefly(f,\{(A_{01},A_{10})\},\emptyset,e)$ terminates after the first iteration returning the $(\mathcal{B}\circ e)$-simulation $\{(A_{01},A_{10})\}$. Notably, this means that the algorithm did not need to process any polynomial other than $f_{A_{01}}$ and $f_{A_{10}}$. Note also that $e(\{(A_{01},A_{10})\})$ is the BDE from Example~\ref{example:002}.
\end{examplenum}

%

An immediate consequence of Lemma~\ref{lem:B-compatible} and Definition~\ref{def:FBD} is the $\cal{F}$-compatibility of reflexive, symmetric, transitive, and equivalence closures which can be employed in Algorithm~\ref{alg:findFBD}.
\begin{corollary}
$r$, $s$, $t$, and $e$ are $\cal{F}$-compatible up-to techniques.
\end{corollary}

Without constraints we can employ a number of up-to techniques. Unfortunately, the up-to techniques considered above turn out not to be $\cal{B}_C$-sound for generic choices of the constraints $C \subseteq \X \times \X$, as illustrated in the following example.
\begin{examplenum}
Consider the following linear vector field $f$ over $X = \{x, x', x'', y, y', y'', z, z', z''\}$ defined by
\begin{align*}
& f_x = 0 \,,&& f_y = x \,, && f_z = y \,, \\
& f_{x'} = 0 \,, && f_{y'} = x' \,,&& f_{z'} = y' \,, \\
& f_{x''} = 0 \,, && f_{y''} = x'' \,,&& f_{z''} = y'' \,.
\end{align*}
Fix $C = \{(x,x'')\}$. We have that neither $s$ nor $t$ are $\mathcal{B}_C$-sound. Indeed $(z,z'') \not\in \gfp{\mathcal{B}_C}$ but both $\gfp{\mathcal{B}_C \circ s}$ and $\gfp{\mathcal{B}_C \circ t}$ contain $(z,z'')$.
This is also a counterexample for the $\mathcal{B}_C$-soundness of $e$.
\end{examplenum}

Despite this fact, under sensible conditions on $C$, reflexive and symmetric closures are still sound up-to techniques. 
\begin{lemma} \label{lem:reflupto}
If $C \cap \mathtt{id} = \emptyset$, then $r$ is $\mathcal{B}_{C}$-compatible.
\end{lemma}

The following result states that when the relation $C \subseteq \X \times \X$ symmetric, then the symmetric closure is $\cal{B}_C$-compatible.
\begin{lemma}\label{lem:BsC-compatible}
Let $C \subseteq \X \times \X$. Then, $s$ is $\mathcal{B}_{s(C)}$-compatible.
\end{lemma}


\begin{remark}
\ch{Up-to techniques are particularly useful in discovering if the backward equivalence of two variables $x$ and $y$ is dependent on the equivalence of some other (pairwise-distinct) pairs of variables $C \subseteq (\X \times \X)$. Indeed, this can be done by checking if $\onthefly(f, \{(x,y)\}, s(C), r \circ s) = \emptyset$.}
\end{remark}

\section{Applications} \label{sec:applications}
We present applications which show how our \ch{local} algorithm equipped with up-to-techniques and constraints is complementary to one based on partition refinement for the computation of differential equivalences.
Section~\ref{sec:curried} discusses how the \ch{local} algorithm can be used to find relations between ODE variables that are independent of the choice of model parameters. 
%
Section~\ref{sec:ei} shows that the up-to techniques can speed up BDB computations by several orders of magnitude.
Section~\ref{sec:bn} discusses how the use of constraints allows one to compute BDBs satisfying domain-specific properties not supported by partition refinement. 
Section~\ref{sec:otflump} shows how the \ch{local} algorithm disproves equivalence
by exploring only a small part of the model, thus outperforming the partition refinement approach.
Section~\ref{sec:sis} computes FDBs of epidemiological models on graphs.

Results are based on a prototype implementation of Algorithm~\ref{alg:otf}; the computations of BDE and FDE were performed with the tool ERODE~\cite{tacas17cttv}, run on a common laptop with 3.1 GHz Dual-Core Intel Core i5 and 8GB RAM.


\setlength{\tabcolsep}{3.1pt}
\begin{table*}[t]
\centering
\scalebox{0.93}{
\begin{tabular}[t]{c rr H lrl lrl lrl lrl lrl}
\toprule
\multicolumn{3}{c}{\emph{Model}} &&
\multicolumn{3}{c}{\emph{Base algorithm}}&\multicolumn{3}{c}{\emph{Up-to reflexivity}}&\multicolumn{3}{c}{\emph{Up-to symmetry}}&\multicolumn{3}{c}{\emph{Up-to transitivity}}&\multicolumn{3}{c}{\emph{Up-to equivalence}}
\\
\cmidrule(l{2pt}r{2pt}){1-3} \cmidrule(l{2pt}r{2pt}){5-7} \cmidrule(l{2pt}r{2pt}){8-10} \cmidrule(l{2pt}r{2pt}){11-13} \cmidrule(l{2pt}r{2pt}){14-16} \cmidrule(l{2pt}r){17-19}
\multicolumn{1}{c}{$n$} & \multicolumn{1}{c}{$|\mathcal{S}|$} & \multicolumn{1}{c}{$|\mathcal{R}|$} &
\multicolumn{1}{H}{$|\mathit{Q}|$}
& \multicolumn{1}{c}{\emph{Time (s)}} & \multicolumn{1}{c}{\emph{TP}} & \multicolumn{1}{c}{$|\R|$}
& \multicolumn{1}{c}{\emph{Time (s)}} & \multicolumn{1}{c}{\emph{TP}} & \multicolumn{1}{c}{$|\R|$}
& \multicolumn{1}{c}{\emph{Time (s)}} & \multicolumn{1}{c}{\emph{TP}} & \multicolumn{1}{c}{$|\R|$}
& \multicolumn{1}{c}{\emph{Time (s)}} & \multicolumn{1}{c}{\emph{TP}} & \multicolumn{1}{c}{$|\R|$}
& \multicolumn{1}{c}{\emph{Time (s)}} & \multicolumn{1}{c}{\emph{TP}} & \multicolumn{1}{c}{$|\R|$}
\\
\cmidrule(l{2pt}r{2pt}){1-3} \cmidrule(l{2pt}r{2pt}){4-4} \cmidrule(l{2pt}r{2pt}){5-7} \cmidrule(l{2pt}r{2pt}){8-10} \cmidrule(l{2pt}r{2pt}){11-13} \cmidrule(l{2pt}r{2pt}){14-16} \cmidrule(l{2pt}r){17-19}
$2$ & \snt{18} & \snt{48} & 4
& \worst{8.00E-3} & \snt{946} & \snt{30} 
& 4.00E-3 & \snt{432} & \snt{14} 
& 7.50E-3 & \snt{839} & \snt{25} 
& 8.00E-3 & \snt{926} & \snt{29} 
& \best{3.00E-3} & \snt{323} & 9.00E0 
\\
$3$ & \snt{66} & \snt{288} & 9
& \worst{1.00E-1} & \snt{26771} & \snt{234} 
& 6.50E-2 & \snt{17453} & \snt{171} 
& 7.40E-2 & \snt{19596} & \snt{153} 
& 3.70E-2 & \snt{9104} & \snt{157} 
& \best{1.50E-2} & \snt{2940} & \snt{51} 
\\
$4$ & \snt{258} & \snt{1536} & 16
& \worst{6.70E-1} & \snt{308928} & \snt{1770} 
& 5.20E-1 & \snt{245328} & \snt{1516} 
& 4.90E-1 & \snt{200591} & \snt{1020} 
& 1.70E-1 & \snt{66915} & \snt{763} 
& \best{3.90E-2} & \snt{19850} & \snt{234} 
\\
$5$ & \snt{1026} & \snt{7680} & 25
& \worst{6.90E0} & \snt{3230798} & \snt{13386} 
& 5.40E0 & \snt{2853023} & \snt{12365} 
& 4.23E0 & \snt{1947985} & \snt{7216} 
& 1.91E0 & \snt{419242} & \snt{3449} 
& \best{4.10E-1} & \snt{114578} & \snt{989} 
\\
$6$ & \snt{4098} & \snt{36864} & 36
& \worst{9.30E1} & \snt{31607119} & \snt{99450} 
& 6.80E1 & \snt{29537884} & \snt{95358} 
& 5.91E1 & \snt{18242790} & \snt{51789} 
& 2.71E1 & \snt{2392072} & \snt{14997} 
& \best{4.00E0} & \snt{608325} & \snt{4043} 
\\
\bottomrule
 \end{tabular}
}
\caption{Speed-up of BDB computation with up-to techniques. We highlight in bold the best and worst runtimes.
}
\label{table:ei}
\end{table*}
\setlength{\tabcolsep}{6pt} 

\subsection{Parameter-independent bisimulations}\label{sec:curried}
Building mechanistic models in biology is hindered by the difficulty in observing all biochemical interactions~\cite{doi:10.1098/rsif.2017.0237},  which may lead to uncertainties in the choice of the parameters. When analyzing a model built under these conditions, it is helpful to discover properties that depend only on the \emph{structure} and not on the specific choice of parameter values. Even if the parameters were precise, discovering structural properties can be beneficial, for instance, to conduct sensitivity analysis for predictive purposes across different parameterizations.

With Example~\ref{example:001}, we discuss how to compute bisimulations that hold independently of the values of some kinetic rate parameters. For this, we consider a variant of the model where the parameters are explicitly treated as ODEs. This is done by building an \emph{extended} model where the kinetic rate parameter of each reaction is treated as a further species in the CRN (we call them \emph{parameter-species} to ease the presentation), and every modified reaction occurs with fixed rate 1. For each reaction \emph{Ri}, we name the parameter-species of the forward and reverse part as $\param{i}$ and $\paramrev{i}$, respectively.
We obtain the extended CRN:
\begin{align*}
\text{\emph{R1}: } A_{00} + B {\color{gray}+\param{1}} \stackrel{1}{\to} A_{10} {\color{gray}+\param{1}}
\qquad
A_{10} {\color{gray}+\paramrev{1}} \stackrel{1}{\to} A_{00} + B  {\color{gray}+ \paramrev{1}}
\\
\text{\emph{R2}: } A_{00} + B {\color{gray}+\param{2}} \stackrel{1}{\to} A_{01} {\color{gray}+\param{2}}
\qquad
A_{01} {\color{gray}+\paramrev{2}} \stackrel{1}{\to} A_{00} + B  {\color{gray}+ \paramrev{2}}
\\
\text{\emph{R3}: } A_{10} + B {\color{gray}+\param{3}} \stackrel{1}{\to} A_{11} {\color{gray}+\param{3}}
\qquad
A_{11} {\color{gray}+\paramrev{3}} \stackrel{1}{\to} A_{10} + B  {\color{gray}+ \paramrev{3}}
\\
\text{\emph{R4}: } A_{01} + B {\color{gray}+\param{4}} \stackrel{1}{\to} A_{11} {\color{gray}+\param{4}}
\qquad
A_{11} {\color{gray}+\paramrev{4}} \stackrel{1}{\to} A_{01} + B  {\color{gray}+ \paramrev{4}}
\end{align*}
By the law of mass action, the derivatives associated with the parameter-species in the extended CRN are 0: thus the ODEs of the original CRN and those of the extended one  coincide when the values of the original parameters are used as
the initial conditions of the corresponding parameter-species. 

On this extended model one could use the partition refinement algorithm from~\cite{pnas17}, starting from the initial partition
\begin{multline*}
\{ \{ A_{01}, A_{10}, A_{11}, B \}, \{\param{1},\param{2}\},\{\param{3},\param{4}\},\{\paramrev{1},\paramrev{2},\paramrev{3},\paramrev{4}\} \}.
\end{multline*}
This equates parameters that have the same values in the original model. The coarsest refinement by BDE equates species $A_{01}$ and $A_{10}$, as in the original model, without splitting the initial blocks of parameter-species. Thus, $A_{01}$ and $A_{10}$ are backward equivalent \emph{as long as all parameters of the reverse reactions $\overline{k}_i$ are equal}.

Checking if the equivalence carries over under fewer constraints for the kinetic rate parameters amounts to asking whether there exists a BDE that refines the blocks of parameter-species. With a partition refinement algorithm, this may involve checking exponentially many initial partitions (i.e., the ones singling out each reverse rate, the ones with 2 blocks of size 2 for the reverse rates, and so on). Instead, by initializing our \ch{local} algorithm with the query $\{ (A_{01},A_{10}) \}$, we obtain a BDB that only relates $\paramrev{1}$ with $\paramrev{2}$ and $\paramrev{3}$ with $\paramrev{4}$, respectively, as well as $\param{1}$ with $\param{2}$, and $\param{3}$ with $\param{4}$. This gives such a desired refinement: the BDB requires one equality for the  reverse rates of reactions \emph{R1}-\emph{R2}, and one for the reverse rates of reactions \emph{R3}-\emph{R4}.

%
The above BDB query required no longer than 70ms and the solution of 103 transportation problems.

\subsection{Computational speed-up using up-to techniques}\label{sec:ei}

To show how  up-to techniques  can speed up the computation of bisimulations,  we consider a benchmark model featuring complex formation by multisite phosphorylation from~\cite[Supplementary Note 7]{citeulike:8493139}; it has been used in \cite{DBLP:conf/tacas/CardelliTTV16,tacas17cttv} to assess the scalability of the partition refinement algorithms for BDE and FDE.
Example~\ref{example:001} can be seen as a simpler version.

Here we consider a family of models of varying size obtained by changing the number of binding sites, $n$, from 2 to 7.
For each $n$, the query encoded the question whether there exists a BDB relating all species representing  complexes where all $n$ binding sites are unphosphorylated; this leads to a query set $Q$ of size $|Q|=n^2$ for each $n$. 

Table~\ref{table:ei} shows the runtimes for BDB computations without up-to techniques and with four different up-to techniques, namely, reflexive, symmetric, transitivive, and equivalence closures.
For each case we measured the runtime of the \ch{local} algorithm (column \emph{Time}), the number of solved transportation problems (column \emph{TP}), and the number of pairs in the relations (column $|\R|$). For each $n$, the table also shows the number of species and reactions of the resulting model (columns $|\mathcal{S}|$ and $|\mathcal{R}|$, respectively).
The use of up-to techniques leads to smaller relations, fewer transportation problems to be solved, and therefore lower runtimes.
Improvements can be found already for $n=2$, where, e.g., the runtime for the up-to equivalence case is about 37\% of the base algorithm case. For larger models, runtimes are up to one order of magnitude smaller.

\subsection{Input-preserving BDB for Boolean networks}\label{sec:bn}
Boolean networks (BN), proposed in 1969~\cite{Kauffman69,THOMAS1973563}, are an established model of biological systems (e.g.,~\cite{FrontGen2016,Pharma2018}). A BN is given by a set of Boolean variables and a Boolean update function associated with each variable. Given an initial state, i.e., an assignment for each variable, a discrete-time dynamics is obtained by setting the new state of each variable as the evaluation of its associated function. 
Figure~\ref{fig:merged} is an excerpt of the graphical representation of the largest BN from 
the GinSim repository~\cite{NALDI2009134}, taken from~\cite{rodriguez2019cooperation}.
Nodes denote variables related to biochemical species (four are highlighted in cyan for the forecoming discussion), while directed edges are drawn if the source node appears in the Boolean update function of the target.

In some cases, a Boolean representation of the state may be too crude an approximation, e.g., one would like to model the level of activity of a gene continuously in the interval
 $[0,1]$. To cope with this, ODE interpretations of BNs have been proposed in the  literature (e.g.~\cite{odification}), 
essentially by interpolating the Boolean function with a real one that agrees with the Boolean output when evaluated with 0/1 inputs. \emph{Odefy} provides such an approximation with polynomial ODEs~\cite{odefy}.  The ODE encoding of the BN in Figure~\ref{fig:merged} has 129 ODEs with derivatives containing  554 monomials of degrees ranging from 1 to 6.

Motivated by the need to reduce complexity of BNs for their analysis~\cite{BMCBioinf2006,BMCBioinf2014,naldi2011dynamically,veliz2011reduction,saadatpour2013reduction}, BDE has been used for the reduction of these ODEs in the past~\cite{pnas17}. Here we show how the use of constraints with our \ch{local} algorithm can provide informative relations that are not directly obtainable using partition refinement.

In applications to gene regulatory networks, certain variables are usually designated as \emph{input} variables, typically representing species at the top of a signalling pathway (and modeled with the identity update function). The modeler then conducts several analyses by varying their initial value and observing the values of \emph{output} variables of interest (e.g., gene expression levels at the bottom of the pathway). In Figure~\ref{fig:merged}, $\mathit{CD6}$ and $\mathit{TLR5}$ are two inputs (out of 16 not shown in the excerpt).

The maximal BDE, obtained by partition refinement using the initial trivial partition with one block only, contains an equivalence class with 20 variables that also includes the four highlighted in Figure~\ref{fig:merged}, and a number of input variables. This suggests a biologically relevant pattern of co-expression whereby a number of variables of the network simultaneously respond with the same output to the same input (e.g.,~\cite{tlgl}). However, from the maximal BDE one cannot conclude that this  pattern does not depend on other network conditions, because the larger equivalence class in which it is found imposes the pre-condition that \emph{all} equivalent species start from the same initial condition, including the other inputs in the class.

Specifically, one would wish to find that this pattern holds independently of the values of the input variables. With partition refinement, this can be done by keeping the ODE variables corresponding to the inputs as distinct singleton blocks of the initial partition, as done, e.g., in~\cite{pnas17}.
However, since inputs are separated from all other variables, the coarsest BDE refinement of the as-constructed partition cannot yield the relation among the  highlighted variables in Figure~\ref{fig:merged} (which, in fact, are separated into 3 distinct blocks).

\begin{figure}[t]
\centering
\includegraphics[width=0.95\linewidth]{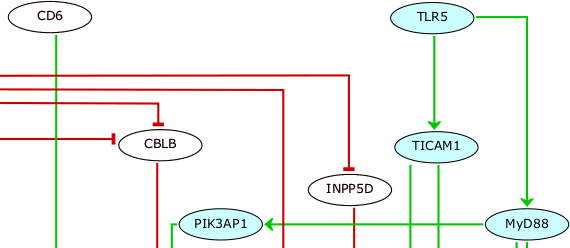}
\caption{Excerpt of the BN TCR-TLR5~\cite{rodriguez2019cooperation} using GinSim~\cite{NALDI2009134}.
}\label{fig:merged}
\end{figure}

We perform three experiments 
 showing how BDB can be successfully used in this context.
%
In order to check whether we can relate the four highlighted species without relating any inputs, or any other species in general, we run our algorithm with query 
the Cartesian product of the 4 highlighted species, and with constraints the Cartesian product of the 16 input species minus the identity (240 pairs).
Our prototype computed a BDB with 16 pairs relating all highlighted species. 
The BDE obtained as the equivalence closure of such BDB contains a block consisting of the  highlighted species, and 1 singleton block for each other species, proving that no further assumptions are required to relate the  highlighted species.
%


As a second example, we use as query 
the Cartesian product of the 3 non-input highlighted species, obtaining a BDB with 12 pairs. 
The BDB again relates all highlighted species, and its equivalence closure gives the same BDE as in the previous experiment.
This suggests that in order to prove the query we need to relate all  highlighted species.
%

In order to confirm this, we perform a third experiment where we explicitly add constraints for preventing $\mathit{TLR5}$ to be related with the other highlighted species.
We get a BDB with 6 pairs only, relating only $\mathit{MyD88}$ with $\mathit{TICAM1}$.
This confirms that it is possible to relate
$\mathit{PIK3AP1}$ with $\mathit{TICAM1}$ and $\mathit{MyD88}$ only if $\mathit{TLR5}$ is also related to them.

All  experiments required less than 0.5s and the solution of less than 85 transportation problems

We note that the first experiment could be replicated with partition refinement using as initial partition the computed BDE, while the second and third one could not because partition refinement can only decrease the size of relations while our \ch{local} approach adds the necessary pairs.

\subsection{Local computation of CTMC lumpability}\label{sec:otflump}
In this section we study the effectiveness of our \ch{local} algorithm with respect to partition refinement in proving that two CTMCs are not equivalent.  We show this on a case study of comparison between a model of a mutual-inhibition (MI) cell switch, a mechanism present in many biological networks (e.g.~\cite{mione,mitwo}), and a simpler switch, which has been shown to correspond to the approximate majority (AM) algorithm from population protocols~\cite{am}. These two systems can be modeled as CRNs as follows:
\begin{align*}
\begin{aligned}[c]
y_0 + z_0 &  \! \stackrel{1}{\to} \!  z_0 + y_1 & y_1 + z_0 &  \! \stackrel{1}{\to} \!  z_0 + y_2 \\
y_2 + y_0 &  \! \stackrel{1}{\to} \!  y_0 + y_1 & y_1 + y_0 &  \! \stackrel{1}{\to} \!  y_0 + y_0 \\
z_2 + z_0 &  \! \stackrel{1}{\to} \!  z_0 + z_1 & z_1 + z_0 &  \! \stackrel{1}{\to} \!  z_0 + z_0 \\
z_0 + y_0 &  \! \stackrel{1}{\to} \!  y_0 + z_1 & z_1 + y_0 &  \! \stackrel{1}{\to} \!  y_0 + z_2
\end{aligned}
\ \ \vline \ \
\begin{aligned}
 x_0 + x_2 &  \! \stackrel{1}{\to} \!  x_2 + x_1 \\
 x_1 + x_2 &  \! \stackrel{1}{\to} \!  x_2 + x_2 \\
 x_2 + x_0 &  \! \stackrel{1}{\to} \!  x_0 + x_1 \\
x_1 + x_0 &  \! \stackrel{1}{\to} \!  x_0 + x_0
\end{aligned}
\end{align*}
where species $y_0$, $y_1$, $y_2$, $z_0$, $z_1$ and $z_2$ refer to MI and $x_0$, $x_1$, $x_2$ refer to AM.  Being able to compare such kinds of networks is relevant in evolutionary biology to assess, for example, if complex cellular mechanisms can be related to less robust, more primordial variants (e.g.,~\cite{CardelliCRN,Gay:2010aa}).

\setlength{\tabcolsep}{4pt}
\begin{table}
\centering
\scalebox{1}{
\begin{tabular}[t]{cc rr l lHH}
\toprule
\multicolumn{2}{c}{\emph{Initial populations}} &
\multicolumn{2}{c}{\emph{Union CTMC}} &
\multicolumn{1}{c}{\emph{BDE}} &
\multicolumn{3}{c}{\emph{BDB}}
\\
\cmidrule(l{2pt}r{2pt}){1-2} \cmidrule(l{2pt}r{2pt}){3-4} \cmidrule(l{2pt}r{2pt}){5-5} \cmidrule(l{2pt}r{2pt}){6-8}
\multicolumn{1}{c}{$x_0,y_0,z_2$} & \multicolumn{1}{c}{$x_2,y_2,z_0$} &
\multicolumn{1}{c}{\emph{States}} & \multicolumn{1}{c}{\emph{Trans.}} &
\multicolumn{1}{c}{\emph{Time (s)}} &
\multicolumn{1}{c}{\emph{Time (s)}} &
\multicolumn{1}{H}{\emph{TP}} &
\multicolumn{1}{H}{$|\mathcal{R}|$}
\\
\midrule
{\color{white}0}2 & {\color{white}0}1 & 93 & 276 & 1.60E-3 &
1.00E-3&\np{1}&0
\\
{\color{white}0}4 & {\color{white}0}2 & 762 & \np{3504} & 7.00E-3 &
1.00E-3&\np{1}&\np{0}
\\
{\color{white}0}6 & {\color{white}0}3 & \np{2979} & \np{16164} & 2.60E-2 &
1.00E-3&\np{1}&\np{0}
\\
{\color{white}0}8 & {\color{white}0}4 & \np{8202} & \np{48624} & 8.90E-2 &
2.00E-3&\np{1}&\np{0}
\\
10 & {\color{white}0}5 & \np{18375} & \np{115140} & 2.89E-1 &
2.00E-3&\np{1}&\np{0}
\\
12 & {\color{white}0}6 & \np{35928} & \np{233856} & 9.10E-1 &
2.00E-3&\np{1}&\np{0}
\\
14 & {\color{white}0}8 & \np{75922} & \np{511984} & {1.18E0} &
{2.00E-3}&\np{1}&\np{0}
\\
16 & 10 & \np{142532} & \np{985504} & {3.10E0} &
3.00E-3&\np{1}&\np{0}
\\
18 & 12 & \np{245550} & \np{1729680} & {5.90E0} &
3.00E-3&\np{1}&\np{0}
\\
20 & 14 & \np{396304} & \np{2832064} & 1.12E1 &
{4.50E-3}&\np{1}&\np{0}
\\
22 & 16 & \np{607658} & \np{4392496} & 1.70E1 &
{6.00E-3}&\np{1}&\np{0}
\\
\bottomrule
 \end{tabular}
}
\caption{Local vs global BDE checks to relate the CTMCs of MI and AM (initial populations not shown are set to 0).
}
\label{table:AMMICTMC}
\end{table}
\setlength{\tabcolsep}{6pt}
It has been established that the species of MI and AM can be related by a BDE when the CRNs are interpreted with ODE mass-action semantics~\cite{DBLP:journals/tcs/CardelliTTV19,DBLP:journals/tcs/CardelliTTV19a}.
Such interpretation can be seen as a deterministic limit description of a CTMC describing discrete molecular interactions, when the number of molecules goes to infinity~\cite{Gillespie77}. In fact, under certain physical conditions, from first principles it is well-known that the ground-truth behavior is given by a CTMC~\cite{Gillespie77}. Starting from an initial population represented by a vector where each component models the amount of elements for each species, 
the CTMC is generated by exhaustively applying every reaction, generating a new CTMC state where the reaction's reagents 
are replaced by the reaction's products, respectively referring to the species appearing in the left- and right-hand-side of the reaction.

%
To see whether MI and AM can be related also under this CTMC interpretation, i.e., under the assumption of finite populations for the species, it is possible to initialize the \ch{local} algorithm with the pair consisting of the respective initial states of the two CTMCs.  Similarly, the partition refinement algorithm can be initialized with a partition with a block containing the two initial states.

The answer to the above comparison question is negative.  Table~\ref{table:AMMICTMC} shows the runtimes of the \ch{local} and the (global) partition-refinement algorithms to disprove the relation between the initial states of the networks' CTMCs with varying initial populations as specified by the first two columns (species not mentioned are set to 0). Increasing initial population counts leads to a combinatorial explosion of the underlying CTMC: the third and fourth columns provide the number of states and transitions, respectively, of the CTMC formed by the disjoint union of the CTMCs of both networks. While the runtime for partition refinement (column \emph{BDE}) grows with the CTMC size (as expected by the computational complexity of CTMC lumping algorithms~\cite{DBLP:journals/ipl/DerisaviHS03}), the \ch{local} algorithm always yields an empty relation by analyzing only one transportation problem within a few milliseconds (column \emph{BDB}).


\subsection{Local FDB computation on epidemiological models}\label{sec:sis}

\begin{figure}
\centering
\includegraphics[width=1.0\linewidth]{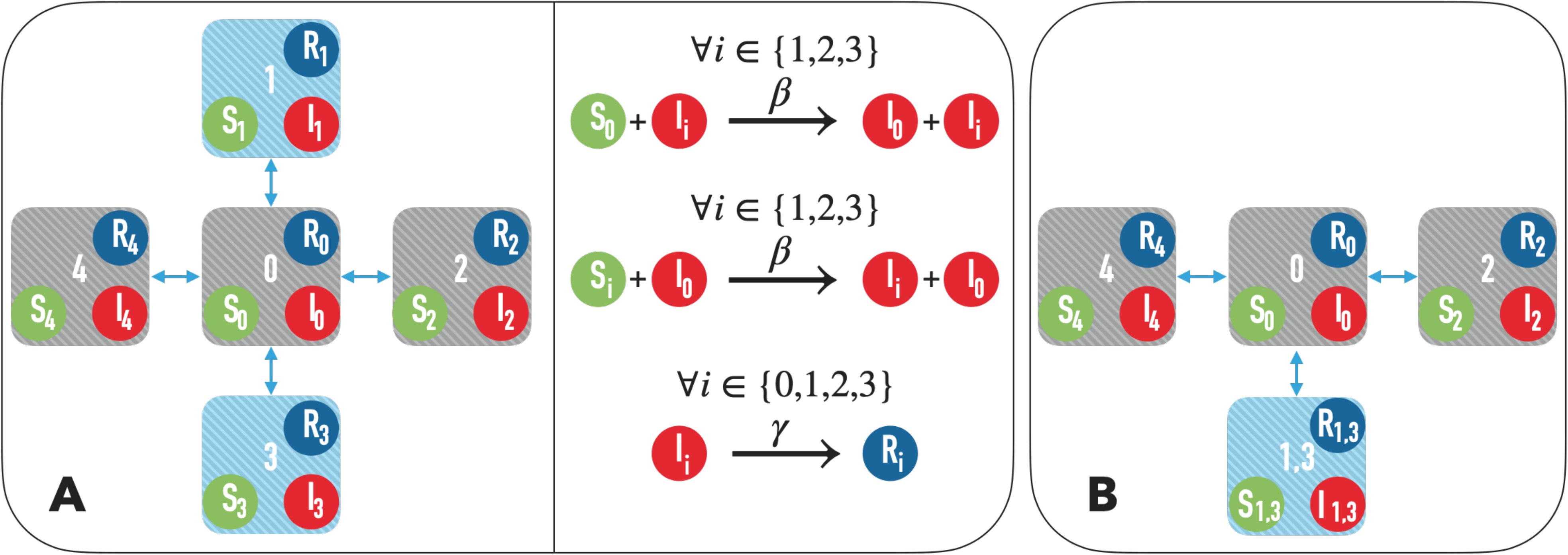}
\caption{
(A) SIR on a star network where species in the central node interact with those in the peripheral ones.
(B) FDB reduction of (A) relating pairwise the variables in nodes 1 and 3
de facto obtaining a coarse-grained 
network.
}
\label{wfig:SIS}
\end{figure}

We study how FDB can be applied on epidemiological SIR models evolving on graphs (e.g.,~\cite{pastor2015epidemic}). In this model, each vertex describes an individual that: can be susceptible (S) to an infection with a rate proportional to the number of
infected (I) neighbors (and parameter $\beta$); that can recover (R) from the infection to gain immunity (with parameter $\gamma$).
%
Fig.~\ref{wfig:SIS}(A) shows the model for a star topology.
%

Since nodes in the graph may represent locations, a coarse-grained representation may be helpful to reason about the system in terms of larger areas.
For example, the cumulative dynamics of nodes 1 and 3  could be obtained by running the FDB \ch{local} algorithm with the query $\{S_1,S_3\}\times\{S_1,S_3\}$.
By solving 295 transportation problems in about 10\,ms the output FDB has 21 pairs: a reflexive pair for each of the 15 species, and the 6 pairs
$(S_1,S_3), (S_3,S_1), (I_1,I_3), (I_3,I_1), (R_1,R_3), (R_3,R_1)$.
Intuitively, as depicted in Fig.~\ref{wfig:SIS}(B), it does not distinguish among nodes 1 and 3, while leaving unaltered the rest of the graph.
%

This problem can be semi-automatically addressed using partition refinement, but it requires to perform trial-and-error for the appropriate initial partition to use.
%
%
E.g., for an initial partition where the only non-singleton block is $\{S_1,S_3\}$, the coarsest FDE refinement is
$\{\{S_0\}$, $\{I_0\}$, $\{S_1,S_3\}$, $\{S_2,S_4\}$, $\{I_0,I_1,I_2,I_3\}$, $\{R_0,R_1,R_2,R_3\}\}$, but it aggregates too much. 
The desired reduction is obtained for initial partition $\{S_1,S_3\}$, $\{I_1,I_3\}$, $\{R_1,R_3\}$, $\{S_0,S_2,S_4,I_0,I_2,I_4,R_0\}$, $\{R_2\}$, $\{R_4\}$.

\section{Conclusions and Future Work}
\ch{We have presented an algorithm for computing bisimulations over variables of a system of polynomial differential equations using a \ch{local} approach, which complements available \ch{global} methods based on partition refinements.
Crucial to our approach is the introduction of a novel coupling method for reasoning about equivalences over polynomials.
Our algorithm computes bisimulations which relate ODE solutions \emph{exactly}. Given that probabilistic couplings have proved instrumental for the development of \emph{bisimilarity distances}~\cite{BreugelW01,DesharnaisGJP04}, a natural question is whether our approach can be lifted to metric spaces for reasoning about approximate bisimulations of polynomial differential equations. We intend to tackle this in future work.}

\paragraph*{\textbf{Acknowledgments}} This work was supported by the Poul Due Jensen Foundation, grant 883901. 

\balance
\bibliography{onthefly}

\begin{thebibliography}{10}
\providecommand{\url}[1]{#1}
\csname url@samestyle\endcsname
\providecommand{\newblock}{\relax}
\providecommand{\bibinfo}[2]{#2}
\providecommand{\BIBentrySTDinterwordspacing}{\spaceskip=0pt\relax}
\providecommand{\BIBentryALTinterwordstretchfactor}{4}
\providecommand{\BIBentryALTinterwordspacing}{\spaceskip=\fontdimen2\font plus
\BIBentryALTinterwordstretchfactor\fontdimen3\font minus
  \fontdimen4\font\relax}
\providecommand{\BIBforeignlanguage}[2]{{%
\expandafter\ifx\csname l@#1\endcsname\relax
\typeout{** WARNING: IEEEtran.bst: No hyphenation pattern has been}%
\typeout{** loaded for the language `#1'. Using the pattern for}%
\typeout{** the default language instead.}%
\else
\language=\csname l@#1\endcsname
\fi
#2}}
\providecommand{\BIBdecl}{\relax}
\BIBdecl

\bibitem{DBLP:conf/concur/BortolussiH12}
L.~Bortolussi and J.~Hillston, ``{Fluid Model Checking},'' in \emph{{CONCUR}},
  2012, pp. 333--347.

\bibitem{Hillston_QEST05}
J.~Hillston, ``Fluid flow approximation of {PEPA} models,'' in \emph{QEST},
  Sep. 2005, pp. 33--43.

\bibitem{biopepa}
F.~Ciocchetta and J.~Hillston, ``{Bio-PEPA}: A framework for the modelling and
  analysis of biological systems,'' \emph{Theoretical Computer Science}, vol.
  410, no. 33-34, pp. 3065--3084, 2009.

\bibitem{lics16}
L.~Cardelli, M.~Tribastone, M.~Tschaikowski, and A.~Vandin, ``Comparing
  chemical reaction networks: A categorical and algorithmic perspective,'' in
  \emph{Proceedings of the Thirty-First Annual ACM/IEEE Symposium on Logic in
  Computer Science (LICS)}, 2016.

\bibitem{popl16}
------, ``Symbolic computation of differential equivalences,'' in
  \emph{{POPL}}, 2016, pp. 137--150.

\bibitem{DBLP:conf/hybrid/Boreale18}
M.~Boreale, ``Algorithms for exact and approximate linear abstractions of
  polynomial continuous systems,'' in \emph{{HSCC}}, M.~Prandini and J.~V.
  Deshmukh, Eds.\hskip 1em plus 0.5em minus 0.4em\relax {ACM}, 2018, pp.
  207--216.

\bibitem{DBLP:journals/scp/Boreale20}
------, ``Complete algorithms for algebraic strongest postconditions and
  weakest preconditions in polynomial odes,'' \emph{Sci. Comput. Program.},
  vol. 193, p. 102441, 2020.

\bibitem{DBLP:journals/tac/PappasLS00}
G.~J. Pappas, G.~Lafferriere, and S.~Sastry, ``Hierarchically consistent
  control systems,'' \emph{{IEEE} Trans. Automat. Contr.}, vol.~45, no.~6, pp.
  1144--1160, 2000.

\bibitem{bisimulation_lin_sys_Schaft}
A.~J. van~der Schaft, ``Equivalence of dynamical systems by bisimulation,''
  \emph{IEEE Transactions on Automatic Control}, vol.~49, 2004.

\bibitem{LiRabitz1997}
J.~Toth, G.~Li, H.~Rabitz, and A.~S. Tomlin, ``The effect of lumping and
  expanding on kinetic differential equations,'' \emph{SIAM Journal on Applied
  Mathematics}, vol.~57, no.~6, pp. 1531--1556, 1997.

\bibitem{okino1998}
M.~S. Okino and M.~L. Mavrovouniotis, ``Simplification of mathematical models
  of chemical reaction systems,'' \emph{Chemical Reviews}, vol.~2, no.~98, pp.
  391--408, 1998.

\bibitem{Larsen19911}
K.~G. Larsen and A.~Skou, ``Bisimulation through probabilistic testing,''
  \emph{Inf. Comput.}, vol.~94, no.~1, pp. 1--28, 1991.

\bibitem{1703385}
J.~Sproston and S.~Donatelli, ``Backward bisimulation in {Markov} chain model
  checking,'' \emph{{IEEE} Trans. Software Eng.}, vol.~32, no.~8, pp. 531--546,
  2006.

\bibitem{BuchholzOrdinaryExact}
P.~Buchholz, ``Exact and ordinary lumpability in finite {M}arkov chains,''
  \emph{Journal of Applied Probability}, vol.~31, no.~1, pp. 59--75, 1994.

\bibitem{DBLP:journals/ipl/DerisaviHS03}
S.~Derisavi, H.~Hermanns, and W.~H. Sanders, ``Optimal state-space lumping in
  markov chains,'' \emph{Inf. Process. Lett.}, vol.~87, no.~6, pp. 309--315,
  2003.

\bibitem{DBLP:conf/tacas/ValmariF10}
A.~Valmari and G.~Franceschinis, ``Simple ${O}(m \log n)$ time {M}arkov chain
  lumping,'' in \emph{{TACAS}}, 2010, pp. 38--52.

\bibitem{pnas17}
L.~Cardelli, M.~Tribastone, M.~Tschaikowski, and A.~Vandin, ``Maximal
  aggregation of polynomial dynamical systems,'' \emph{Proceedings of the
  National Academy of Sciences}, vol. 114, no.~38, pp. 10\,029--10\,034, 2017.

\bibitem{tlgl}
R.~Zhang, M.~V. Shah, J.~Yang, S.~B. Nyland, X.~Liu, J.~K. Yun, R.~Albert, and
  J.~Loughran, Thomas~P, ``Network model of survival signaling in large
  granular lymphocyte leukemia,'' \emph{Proceedings of the National Academy of
  Sciences of the United States of America}, vol. 105, no.~42, pp.
  16\,308--16\,313, 10 2008.

\bibitem{ZhangH05}
B.~Zhang and S.~Horvath, ``{A General Framework for Weighted Gene Co-Expression
  Network Analysis},'' \emph{Statistical Applications in Genetics and Molecular
  Biology}, vol.~4, no.~1, 2005.

\bibitem{Voit:2013aa}
E.~O. Voit, ``Biochemical systems theory: A review,'' \emph{ISRN
  Biomathematics}, vol. 2013, p.~53, 2013.

\bibitem{DBLP:conf/fm/0009ZZZ15}
J.~Liu, N.~Zhan, H.~Zhao, and L.~Zou, ``Abstraction of elementary hybrid
  systems by variable transformation,'' in \emph{{FM}}, vol. 9109, 2015, pp.
  360--377.

\bibitem{LindvallBook}
T.~Lindvall, \emph{Lectures on the Coupling Method}, ser. Wiley Series in
  Probability and Mathematical Statistics.\hskip 1em plus 0.5em minus
  0.4em\relax New York: John Wiley, 1992.

\bibitem{Thorisson95}
H.~Thorisson, ``Coupling methods in probability theory,'' \emph{Scand. J.
  Stat.}, vol.~22, pp. 159--182, 1995.

\bibitem{Dantzig51}
G.~B. Dantzig, ``{Application of the Simplex method to a transportation
  problem},'' in \emph{Activity analysis of production and allocation},
  T.~Koopmans, Ed.\hskip 1em plus 0.5em minus 0.4em\relax J. Wiley, New York,
  1951, pp. 359--373.

\bibitem{FordF56}
L.~R. Ford and D.~R. Fulkerson, ``{Solving the Transportation Problem},''
  \emph{Management Science}, vol.~3, no.~1, pp. 24--32, 1956.

\bibitem{Olrin88}
J.~Orlin, ``{A Faster Strongly Polynomial Minimum Cost Flow Algorithm},'' in
  \emph{STOC}, 1988, pp. 377--387.

\bibitem{milnerBook}
R.~Milner, \emph{Communication and Concurrency}.\hskip 1em plus 0.5em minus
  0.4em\relax USA: Prentice-Hall, Inc., 1989.

\bibitem{PousS12}
D.~Pous and D.~Sangiorgi, ``Enhancements of the bisimulation proof method,'' in
  \emph{Advanced Topics in Bisimulation and Coinduction}, ser. Cambridge tracts
  in theoretical computer science.\hskip 1em plus 0.5em minus 0.4em\relax
  Cambridge University Press, 2012, vol.~52, pp. 233--289.

\bibitem{BonchiGGP18}
F.~Bonchi, P.~Ganty, R.~Giacobazzi, and D.~Pavlovic, ``Sound up-to techniques
  and complete abstract domains,'' in \emph{{LICS}}, 2018, pp. 175--184.

\bibitem{BonchiP13}
F.~Bonchi and D.~Pous, ``Checking {NFA} equivalence with bisimulations up to
  congruence,'' in \emph{{POPL}}.\hskip 1em plus 0.5em minus 0.4em\relax {ACM},
  2013, pp. 457--468.

\bibitem{Bonchi0K17}
F.~Bonchi, B.~K{\"{o}}nig, and S.~K{\"{u}}pper, ``Up-to techniques for weighted
  systems,'' in \emph{TACAS}, 2017, pp. 535--552.

\bibitem{CardelliCRN}
L.~Cardelli, ``Morphisms of reaction networks that couple structure to
  function,'' \emph{BMC Systems Biology}, vol.~8, no.~1, p.~84, 2014.

\bibitem{DBLP:journals/lmcs/Boreale19}
M.~Boreale, ``Algebra, coalgebra, and minimization in polynomial differential
  equations,'' \emph{Log. Methods Comput. Sci.}, vol.~15, no.~1, 2019.

\bibitem{Buchberger76}
B.~Buchberger, ``A theoretical basis for the reduction of polynomials to
  canonical forms,'' \emph{SIGSAM Bull.}, vol.~10, no.~3, pp. 19--29, 1976.

\bibitem{DBLP:conf/tacas/GhorbalP14}
K.~Ghorbal and A.~Platzer, ``Characterizing algebraic invariants by
  differential radical invariants,'' in \emph{{TACAS}}, E.~{\'{A}}brah{\'{a}}m
  and K.~Havelund, Eds., vol. 8413.\hskip 1em plus 0.5em minus 0.4em\relax
  Springer, 2014, pp. 279--294.

\bibitem{DBLP:conf/popl/SankaranarayananSM04}
S.~Sankaranarayanan, H.~Sipma, and Z.~Manna, ``Non-linear loop invariant
  generation using {G}r{\"{o}}bner bases,'' in \emph{{POPL}}, 2004, pp.
  318--329.

\bibitem{DBLP:conf/lics/Platzer12}
A.~Platzer, ``Logics of dynamical systems,'' in \emph{{LICS}}.\hskip 1em plus
  0.5em minus 0.4em\relax {IEEE} Computer Society, 2012, pp. 13--24.

\bibitem{DBLP:conf/lics/JonssonL91}
B.~Jonsson and K.~G. Larsen, ``Specification and refinement of probabilistic
  processes,'' in \emph{{LICS}}, 1991, pp. 266--277.

\bibitem{DBLP:conf/cav/Baier96}
C.~Baier, ``Polynomial time algorithms for testing probabilistic bisimulation
  and simulation,'' in \emph{{CAV}}, ser. Lecture Notes in Computer Science,
  R.~Alur and T.~A. Henzinger, Eds., vol. 1102.\hskip 1em plus 0.5em minus
  0.4em\relax Springer, 1996, pp. 50--61.

\bibitem{DBLP:conf/fossacs/ChenBW12}
D.~Chen, F.~van Breugel, and J.~Worrell, ``On the complexity of computing
  probabilistic bisimilarity,'' in \emph{{FOSSACS}}, 2012, pp. 437--451.

\bibitem{Bacci:2013}
G.~Bacci, G.~Bacci, K.~G. Larsen, and R.~Mardare, ``On-the-fly exact
  computation of bisimilarity distances,'' in \emph{TACAS}, 2013.

\bibitem{BartheEGHSS15}
G.~Barthe, T.~Espitau, B.~Gr{\'{e}}goire, J.~Hsu, L.~Stefanesco, and P.~Strub,
  ``Relational reasoning via probabilistic coupling,'' in \emph{{LPAR}}, ser.
  Lecture Notes in Computer Science, vol. 9450.\hskip 1em plus 0.5em minus
  0.4em\relax Springer, 2015, pp. 387--401.

\bibitem{BartheGHS17}
G.~Barthe, B.~Gr{\'{e}}goire, J.~Hsu, and P.~Strub, ``Coupling proofs are
  probabilistic product programs,'' in \emph{{POPL}}, 2017, pp. 161--174.

\bibitem{AguirreBBBG18}
A.~Aguirre, G.~Barthe, L.~Birkedal, A.~Bizjak, M.~Gaboardi, and D.~Garg,
  ``Relational reasoning for {Markov} chains in a probabilistic guarded lambda
  calculus,'' in \emph{{ESOP}}, vol. 10801, 2018, pp. 214--241.

\bibitem{BartheEGHS18}
G.~Barthe, T.~Espitau, B.~Gr{\'{e}}goire, J.~Hsu, and P.~Strub, ``Proving
  expected sensitivity of probabilistic programs,'' \emph{Proc. {ACM} Program.
  Lang.}, vol.~2, no. {POPL}, pp. 57:1--57:29, 2018.

\bibitem{BartheFGGHS16}
G.~Barthe, N.~Fong, M.~Gaboardi, B.~Gr{\'{e}}goire, J.~Hsu, and P.~Strub,
  ``Advanced probabilistic couplings for differential privacy,'' in
  \emph{{CCS}}.\hskip 1em plus 0.5em minus 0.4em\relax {ACM}, 2016, pp. 55--67.

\bibitem{BalleBG18}
B.~Balle, G.~Barthe, and M.~Gaboardi, ``Privacy amplification by subsampling:
  Tight analyses via couplings and divergences,'' in \emph{NeurIPS}, 2018, pp.
  6280--6290.

\bibitem{BartheEGHS17}
G.~Barthe, T.~Espitau, B.~Gr{\'{e}}goire, J.~Hsu, and P.~Strub, ``Proving
  uniformity and independence by self-composition and coupling,'' in
  \emph{{LPAR}}, vol.~46, 2017, pp. 385--403.

\bibitem{concur15}
L.~Cardelli, M.~Tribastone, M.~Tschaikowski, and A.~Vandin, ``Forward and
  backward bisimulations for chemical reaction networks,'' in \emph{CONCUR},
  2015, pp. 226--239.

\bibitem{Strassen}
V.~Strassen, ``The existence of probability measures with given marginals,''
  \emph{Annals of Mathematical Statistics}, vol.~36, no.~2, pp. 423--439, 04
  1965.

\bibitem{Li19891413}
G.~Li and H.~Rabitz, ``A general analysis of exact lumping in chemical
  kinetics,'' \emph{Chemical Engineering Science}, vol.~44, no.~6, pp.
  1413--1430, 1989.

\bibitem{DBLP:journals/corr/abs-2004-11961}
\BIBentryALTinterwordspacing
A.~Ovchinnikov, I.~C. P{\'{e}}rez{-}Verona, G.~Pogudin, and M.~Tribastone,
  ``{CLUE:} exact maximal reduction of kinetic models by constrained lumping of
  differential equations,'' \emph{CoRR}, vol. abs/2004.11961, 2020. [Online].
  Available: \url{https://arxiv.org/abs/2004.11961}
\BIBentrySTDinterwordspacing

\bibitem{tacas17cttv}
L.~Cardelli, M.~Tribastone, M.~Tschaikowski, and A.~Vandin, ``{ERODE}: A tool
  for the evaluation and reduction of ordinary differential equations,'' in
  \emph{TACAS}, 2017.

\bibitem{doi:10.1098/rsif.2017.0237}
A.~Babtie and M.~Stumpf, ``How to deal with parameters for whole-cell
  modelling,'' \emph{Journal of The Royal Society Interface}, vol.~14, no. 133,
  p. 20170237, 2017.

\bibitem{citeulike:8493139}
M.~W. Sneddon, J.~R. Faeder, and T.~Emonet, ``Efficient modeling, simulation
  and coarse-graining of biological complexity with {NFsim},'' \emph{Nature
  Methods}, vol.~8, no.~2, pp. 177--183, 2011.

\bibitem{DBLP:conf/tacas/CardelliTTV16}
L.~Cardelli, M.~Tribastone, M.~Tschaikowski, and A.~Vandin, ``{Efficient
  Syntax-Driven Lumping of Differential Equations},'' in \emph{{TACAS}}, 2016,
  pp. 93--111.

\bibitem{Kauffman69}
S.~Kauffman, ``Homeostasis and differentiation in random genetic control
  networks,'' \emph{Nature}, vol. 224, no. 5215, pp. 177--178, 1969.

\bibitem{THOMAS1973563}
R.~Thomas, ``Boolean formalization of genetic control circuits,'' \emph{Journal
  of Theoretical Biology}, vol.~42, no.~3, pp. 563--585, 1973.

\bibitem{FrontGen2016}
W.~Abou-Jaoud{\'e}, P.~Traynard, P.~T. Monteiro, J.~Saez-Rodriguez, T.~Helikar,
  D.~Thieffry, and C.~Chaouiya, ``Logical modeling and dynamical analysis of
  cellular networks,'' \emph{Frontiers in Genetics}, vol.~7, pp. 94--94, 05
  2016.

\bibitem{Pharma2018}
P.~Bloomingdale, V.~A. Nguyen, J.~Niu, and D.~E. Mager, ``Boolean network
  modeling in systems pharmacology,'' \emph{Journal of pharmacokinetics and
  pharmacodynamics}, vol.~45, no.~1, pp. 159--180, 02 2018.

\bibitem{NALDI2009134}
A.~Naldi, D.~Berenguier, A.~Faur{\'e}, F.~Lopez, D.~Thieffry, and C.~Chaouiya,
  ``Logical modelling of regulatory networks with {GINsim} 2.3,''
  \emph{Biosystems}, vol.~97, no.~2, pp. 134 -- 139, 2009.

\bibitem{rodriguez2019cooperation}
O.~Rodr{\'\i}guez-Jorge, L.~A. Kempis-Calanis, W.~Abou-Jaoud{\'e}, D.~Y.
  Guti{\'e}rrez-Reyna, C.~Hernandez, O.~Ramirez-Pliego, M.~Thomas-Chollier,
  S.~Spicuglia, M.~A. Santana, and D.~Thieffry, ``Cooperation between {T} cell
  receptor and {Toll-like} receptor 5 signaling for {CD4+ T cell} activation,''
  \emph{Science signaling}, vol.~12, no. 577, 2019.

\bibitem{odification}
D.~M. Wittmann, J.~Krumsiek, J.~Saez-Rodriguez, D.~A. Lauffenburger, S.~Klamt,
  and F.~J. Theis, ``Transforming boolean models to continuous models:
  methodology and application to {T-cell} receptor signaling,'' \emph{BMC
  Systems Biology}, vol.~3, no.~1, p.~98, 2009.

\bibitem{odefy}
J.~Krumsiek, S.~P{\"o}lsterl, D.~M. Wittmann, and F.~J. Theis, ``Odefy - from
  discrete to continuous models,'' \emph{BMC Bioinformatics}, vol.~11, no.~1,
  p. 233, 2010.

\bibitem{BMCBioinf2006}
S.~Klamt, J.~Saez-Rodriguez, J.~A. Lindquist, L.~Simeoni, and E.~D. Gilles, ``A
  methodology for the structural and functional analysis of signaling and
  regulatory networks,'' \emph{BMC Bioinformatics}, vol.~7, no.~1, p.~56, 2006.

\bibitem{BMCBioinf2014}
A.~Veliz-Cuba, B.~Aguilar, F.~Hinkelmann, and R.~Laubenbacher, ``Steady state
  analysis of boolean molecular network models via model reduction and
  computational algebra,'' \emph{BMC Bioinformatics}, vol.~15, no.~1, p. 221,
  2014.

\bibitem{naldi2011dynamically}
A.~Naldi, E.~Remy, D.~Thieffry, and C.~Chaouiya, ``Dynamically consistent
  reduction of logical regulatory graphs,'' \emph{Theoretical Computer
  Science}, vol. 412, no.~21, pp. 2207--2218, 2011.

\bibitem{veliz2011reduction}
A.~Veliz-Cuba, ``Reduction of boolean network models,'' \emph{Journal of
  theoretical biology}, vol. 289, pp. 167--172, 2011.

\bibitem{saadatpour2013reduction}
A.~Saadatpour, R.~Albert, and T.~C. Reluga, ``A reduction method for boolean
  network models proven to conserve attractors,'' \emph{SIAM Journal on Applied
  Dynamical Systems}, vol.~12, no.~4, pp. 1997--2011, 2013.

\bibitem{mione}
A.~Verdugo, P.~K. Vinod, J.~J. Tyson, and B.~Novak, ``Molecular mechanisms
  creating bistable switches at cell cycle transitions,'' \emph{Open Biology},
  vol.~3, no.~3, p. 120179, 03 2013.

\bibitem{mitwo}
T.~S. Gardner, C.~R. Cantor, and J.~J. Collins, ``Construction of a genetic
  toggle switch in {Escherichia} coli,'' \emph{Nature}, vol. 403, no. 6767, pp.
  339--342, 01 2000.

\bibitem{am}
L.~Cardelli and A.~Csik{\'a}sz-Nagy, ``The cell cycle switch computes
  approximate majority,'' \emph{Scientific Reports}, vol.~2, pp. 656 EP --, 09
  2012.

\bibitem{Gay:2010aa}
S.~Gay, S.~Soliman, and F.~Fages, ``A graphical method for reducing and
  relating models in systems biology,'' \emph{Bioinformatics}, vol.~26, no.~18,
  pp. i575--i581, 2010.

\bibitem{DBLP:journals/tcs/CardelliTTV19}
L.~Cardelli, M.~Tribastone, M.~Tschaikowski, and A.~Vandin, ``Comparing
  chemical reaction networks: {A} categorical and algorithmic perspective,''
  \emph{Theor. Comput. Sci.}, vol. 765, pp. 47--66, 2019.

\bibitem{DBLP:journals/tcs/CardelliTTV19a}
------, ``Symbolic computation of differential equivalences,'' \emph{Theor.
  Comput. Sci.}, vol. 777, pp. 132--154, 2019.

\bibitem{Gillespie77}
D.~Gillespie, ``Exact stochastic simulation of coupled chemical reactions,''
  \emph{Journal of Physical Chemistry}, vol.~81, no.~25, pp. 2340--2361,
  December 1977.

\bibitem{pastor2015epidemic}
R.~Pastor-Satorras, C.~Castellano, P.~Van~Mieghem, and A.~Vespignani,
  ``Epidemic processes in complex networks,'' \emph{Reviews of modern physics},
  vol.~87, no.~3, p. 925, 2015.

\bibitem{BreugelW01}
F.~van Breugel and J.~Worrell, ``{Towards Quantitative Verification of
  Probabilistic Transition Systems},'' in \emph{{ICALP}}, 2001, pp. 421--432.

\bibitem{DesharnaisGJP04}
J.~Desharnais, V.~Gupta, R.~Jagadeesan, and P.~Panangaden, ``{Metrics for
  labelled Markov processes},'' \emph{Theor. Comput. Sci.}, vol. 318, no.~3,
  pp. 323--354, 2004.

\end{thebibliography}
		
\onecolumn
\appendix


\section*{Proofs}



\begin{proof}[Proof of Theorem~\ref{thm:linStrassen}]
($\ref{itm:linStrass1} \Rightarrow \ref{itm:linStrass2}$)
Let $v \in \RE^X$ such that $\bigwedge_{(x,y) \in \R} v_x \leq v_y$. By hypothesis there exists $\omega \in \Gamma_\lin(g,h)$ such that $\supp(\omega) \subseteq \R$.
Then, the following hold
\begin{align*}
(g^{+} + h^{-})(v) &= \sum_{x \in X} (g^{+} + h^{-})(x) \cdot v_x \\
&= \sum_{x \in X} \big( \sum_{y \in X} \omega(x,y) \big) \cdot v_x \tag{$\omega \in \Gamma_\lin(g,h)$} \\
&= \sum_{x,y \in X} \omega(x,y) \cdot v_x \\
&= \sum_{(x,y) \in \R} \omega(x,y) \cdot v_x \,. \tag{$\supp(\omega) \subseteq \R$}
\end{align*}
Analogously, $(g^{-} + h^{+})(v) = \sum_{(x,y) \in \R} \omega(x,y) \cdot v_y$. Since for all $(x,y) \in \R$, $v_x \leq v_y$ and $\omega(x,y) \geq 0$,
we have $(g^{+} + h^{-})(v) \leq (g^{-} + h^{+})(v)$, which is equivalent to $g(v) \leq h(v)$.

($\ref{itm:linStrass1} \Rightarrow \ref{itm:linStrass3}$) Follows similarly to the above.

($\ref{itm:linStrass2} \Rightarrow \ref{itm:linStrass1}$)
Define $F$ and $F_1$ as follows
\begin{align*}
F &= \{v \in \RE^X \mid  v_x - v_y \leq 0, \text{ for all $(x,y) \in \R$} \}\\
F_1 &= F \cap \{v \in \RE^X \mid  v_x - v_y \leq 1, \text{ for all $(x,y) \notin \R$} \} \,.
\end{align*}
Note that \ref{itm:linStrass2} is equivalent to $0 \geq \sup_{v \in F}  g(v) - h(v)$.
Moreover, the following inequalities hold
\begin{align*}
\sup_{v \in F}  {} &g(v) - h(v) \\
&\geq \sup_{v \in F_1}  g(v) - h(v) \tag{$F_1 \subseteq F$} \\
&= \sup_{v \in F_1} \sum_{x \in X} \big( (g^{+} + h^{-})(x) - (g^{-} + h^{+})(x) \big) \cdot v_x \\
&= \inf_{\omega \in \Gamma_\lin(g,h)} \sum_{(x,y) \not\in \R} \omega(x,y)
\label{eq:duallin}\tag{*}\\
&\geq 0 \tag{by $\omega(x,y) \geq 0$ for all $x,y \in X$}
\end{align*}
where \eqref{eq:duallin} follows by strong duality for linear programs and the fact that, when $\R$ is an equivalence, $\sup_{v \in F_1} \sum_{x \in X} \big( (g^{+} + h^{-})(x) - (g^{-} + h^{+})(x) \big) \cdot v_x$ is a linear program whose dual is $\inf_{\omega \in \Gamma_\lin(g,h)} \sum_{(x,y) \not\in \R} \omega(x,y)$.

The above implies $0 = \inf_{\omega \in \Gamma_\lin(g,h)} \sum_{(x,y) \not\in \R} \omega(x,y)$ which is equivalent to \ref{itm:linStrass1}.

($\ref{itm:linStrass3} \Rightarrow \ref{itm:linStrass2}$)
It follows by noticing that since $\R$ is symmetric, $\bigwedge_{(x,y) \in \R} v_x = v_y$ is equivalent to $\bigwedge_{(x,y) \in \R} v_x \leq v_y$.
\end{proof}

\begin{proof}[Proof of Theorem~\ref{thm:monStrassen}]
($\ref{itm:monStrass1} \Rightarrow \ref{itm:monStrass2}$)
Let $v \in \RE_{>0}^X$ such that $\bigwedge_{(x,y) \in \R} v_x \leq v_y$. By hypothesis there exists $\rho \in \Gamma_\mon(m,n)$ such that $\supp(\rho) \subseteq \R$. Then, the following hold:
\begin{align*}
m(v) &= \prod_{x \in \X} v_x^{m(x)} \\
&= \prod_{x \in \X} v_x^{\sum_{y \in \X} \rho(x,y)} \tag{$\rho \in \Gamma_\mon(m,n)$} \\
&= \prod_{x,y \in \X} v_x^{\rho(x,y)} \\
&= \prod_{(x,y) \in \R} v_x^{\rho(x,y)} \,. \tag{$\supp(\rho) \subseteq \R$}
\end{align*}
Analogously, $n(v) = \prod_{(x,y) \in \R} v_y^{\rho(x,y)}$. Since for all $(x,y) \in \R$, $v_x \leq v_y$ and $\rho(x,y) \geq 0$ we have $m(v) \leq n(v)$.

($\ref{itm:monStrass1} \Rightarrow \ref{itm:monStrass3}$) Follows similarly to the above.

($\ref{itm:monStrass2} \Rightarrow \ref{itm:monStrass1}$)
Define $F$ and $F_1$ as follows
\begin{align*}
F &= \{v \in \RE_{>0}^X \mid  \log v_x - \log v_y \leq 0, \text{ for all $(x,y) \in \R$} \}  \\
F_1 &= F \cap \{v \in \RE_{>0}^X \mid  \log v_x - \log v_y \leq 1, \text{ for all $(x,y) \notin \R$} \} \,.
\end{align*} 
Recall that for all $a,b \in \RE_{>0}$, $a \leq b$ iff $\log a \leq \log b$. 
Therefore, $F = \{v \in \RE_{>0}^X \mid  v_x - v_y \leq 0, \text{ for all $(x,y) \in \R$} \}$, moreover, 
for all $v \in \RE_{>0}^X$, $m(v) - n(v) \leq 0$ iff $\log m(v) - \log n(v) \leq 0$. 
Hence \ref{itm:monStrass2} is equivalent to $0 \geq \sup_{v \in F}  \log m(v) - \log n(v)$. Moreover, the following inequalities hold
\begin{align*}
\sup_{v \in F} \log m(v) - \log n(v) & \geq \sup_{v \in F_1}  \log m(v) - \log n(v) \tag{$F_1 \subseteq F$} \\
&= \sup_{v \in F_1} \sum_{x \in X} \big( m(x) - n(x) \big) \cdot \log v_x \\
&= \inf_{\rho \in \Gamma_\mon(m,n)} \sum_{(x,y) \not\in \R} \rho(x,y)
\label{dualmon}\tag{*}\\
&\geq 0 \tag{by $\rho(x,y) \geq 0$ for all $x,y \in X$}
\end{align*}
where \eqref{dualmon} follows by strong duality for linear programs. Indeed, $\sup_{v \in F_1} \sum_{x \in X} \big( m(x) - n(x) \big) \cdot \log v_x$ becomes a linear program by applying the change of variable $\log v_x = z_x$ for $x \in X$ with $z \in \RE^X$. Since $\R$ is an equivalence relation, its dual is $\inf_{\rho \in \Gamma_\mon(m,n)} \sum_{(x,y) \not\in \R} \rho(x,y)$.

The above implies $0 = \inf_{\rho \in \Gamma_\mon(m,n)} \sum_{(x,y) \not\in \R} \rho(x,y)$ which is equivalent to \ref{itm:monStrass1}.

($\ref{itm:monStrass3} \Rightarrow \ref{itm:monStrass2}$)
It follows by noticing that since $\R$ is symmetric, $\bigwedge_{(x,y) \in \R} v_x = v_y$ is equivalent to $\bigwedge_{(x,y) \in \R} v_x \leq v_y$. 
\end{proof}

\begin{proof}[Proof of Corollary~\ref{cor:couplingEquiv}]
($\Rightarrow$)
Direct consequence of Theorems~\ref{thm:linStrassen} and \ref{thm:monStrassen}.

($\Leftarrow$) Assume $\R$ is an equivalence. Using Theorem~\ref{thm:monStrassen} one can show that $\mon[\R]$ is an equivalence. Analogously, using Theorem~\ref{thm:linStrassen} one can show that also $\poly[\R]$ is an equivalence.

Fix a representative $x_H \in H$ for each $\R$-equivalence class $H$.
For $p \in \poly[X]$ and  $m \in \mon[X]$, we define
\begin{align*}
p^\R = \sum_{m \in \mon[X]} p(m) \cdot m^\R \,,
&&
m^\R = \prod_{H \in X/\R} x_H^{\sum_{x \in H} m(x)} \,.
\end{align*}
By Theorem~\ref{thm:monStrassen}, $(m, m^\R) \in \mon[\R]$, for all $m \in \mon[X]$. Consequently, by Theorem~\ref{thm:linStrassen}, $(p, p^\R) \in \poly[\R]$ for any $p \in \poly[X]$.

The hypothesis $(\bigwedge_{(x,y) \in \R} v_x = v_y) \Rightarrow p(v) = q(v)$ implies $p^\R(v) = q^\R(v)$,
for all $v \in \RE^X$. Taylor's theorem ensures that $p^\R(m) = q^\R(m)$, for all $m \in \mon[X]$, therefore $p^\R$ and $q^\R$ are identical polynomial expressions. By reflexivity of $\poly[\R]$ we have $(p^\R, q^\R) \in \poly[\R]$. By $(p, p^\R), (q, q^\R) \in \poly[\R]$, symmetry and transitivity of $\poly[\R]$, we conclude $(p,q) \in \poly[\R]$.
\end{proof}

\begin{proof}[Proof of Proposition~\ref{prop:BR}]
By Corollary~\ref{cor:couplingEquiv} and definition of BDB.
\end{proof}

\begin{proof}[Proof of Theorem~\ref{thm:BR2BE}]
\begin{trivlist}
\item
\ref{itm:BR2BE1}
It follows by definition of BDE and the following implications
\begin{align*}
\Big( \bigwedge_{(x,y) \in e(\R)} v_x = v_y \Big)
	\Rightarrow \Big( \bigwedge_{(x,y) \in \R} v_x = v_y \Big)
	\tag{$\R \subseteq e(\R)$ } \\
\Rightarrow \Big( \bigwedge_{(x,y) \in \R} f_x(v) = f_y(v) \Big)
	\tag{Proposition~\ref{prop:BR}} \\
\Rightarrow \Big( \bigwedge_{(x,y) \in e(\R)} f_x(v) = f_y(v) \Big)
	\tag{$=$ equiv.\ rel.}
\end{align*}

\item \ref{itm:BR2BE2} Direct consequence of Corollary~\ref{cor:couplingEquiv}.

\item \ref{itm:BR2BE3} By \ref{itm:BR2BE1}, \ref{itm:BR2BE2}, and Knaster-Tarski fixed-point theorem.
\qedhere
\end{trivlist}
\end{proof}

\begin{proof}[Proof of Theorem~\ref{thm:fe:via:be}]
Let us first assume that $f$ is a linear vector field, i.e., there exists a matrix $J \in \RE^{\X \times \X}$ such that $f(v) = J \cdot v$ for all $v \in \RE^\X$. With this, let matrix $A \in \RE^{\X \times \X}$ be such that its rows constitute the equivalence classes of $\X / \R$ via the relation $a_{x_i,x_j} = 1$ if $(x_i,x_j) \in \R$ and $a_{x_i,x_j} = 0$ otherwise. By~\cite{DBLP:journals/corr/abs-2004-11961}, it follows that $\R$ is a forward equivalence of $J$ if and only if the space spanned by the rows of matrix $A J$ is contained in the space spanned by the rows of matrix $A$. This, in turn, holds true if and only if the space spanned by the columns of matrix $J^T A^T$ is contained in the space spanned by the columns of matrix $A^T$. Since the latter is equivalent to $J^T(U_\R) \subseteq U_\R$, where $U_\R = \{ v \in \RE^\X \mid v_{x_i} = v_{x_j} \text{ for all } (x_i,x_j) \in \R \}$, Theorem 3 of~\cite{popl16} yields the claim. We now drop the assumption of linearity. To this end, we observe that~\cite[Lemma I.1]{DBLP:journals/corr/abs-2004-11961} ensures that $\R$ is an FDE of $f$ if and only if $\R$ is an FDE of each $J_k : \RE^\X \to \RE^\X, v \mapsto J_k v$ from~(\ref{eq:jac:dec}). Thanks to above, this holds true if and only if $\R$ is a BDE of all $J^T_k : \RE^\X \to \RE^\X, v \mapsto J^T_k v$ from~(\ref{eq:jac:dec}).
\end{proof}

\begin{proof}[Proof of Theorem~\ref{thm:fe:on:the:fly}]
All three items are a direct consequence of Definition~\ref{def:FBD}, Theorem~\ref{thm:BR2BE} and Theorem~\ref{thm:fe:via:be}.
\end{proof}

\begin{proof}[Proof of Theorem~\ref{thm:constrBR}]
($\Leftarrow$) Assume $\R \subseteq \mathcal{B}_{C}(\R)$. By definition of $\mathcal{B}_C$, $\mathcal{B}_{C}(\R) \cap C = \emptyset$. Because $\R \subseteq \mathcal{B}_{C}(\R)$ we also have that $\R \cap C = \emptyset$. To prove that $\R$ is a BDB we show that $\R$ is a post-fixed point of $\mathcal{B}$.
\begin{align*}
\R &\subseteq \mathcal{B}_{C}(\R) \tag{hypothesis}\\
&= \mathcal{B}(\R \setminus C) \setminus C  \tag{Eq.~\eqref{eq:BC}} \\
&\subseteq \mathcal{B}(\R \setminus C) \\
&\subseteq \mathcal{B}(\R) \tag{$\R \setminus C \subseteq \R$ and $\mathcal{B}$ monotone}
\end{align*}
($\Rightarrow$) By Definition~\ref{def:constrBDB}, $\R \cap C = \emptyset$ and $\R$ is a BDB. Then \begin{align*}
\R &= \R \setminus C \tag{$\R \cap C = \emptyset$} \\
&\subseteq \mathcal{B}(\R) \setminus C \tag{$\R$ is a BDB} \\
&= \mathcal{B}(\R \setminus C) \setminus C \tag{$\R \cap C = \emptyset$} \\
&= \mathcal{B}_C(\R) \tag{Eq.~\eqref{eq:BC}}
\end{align*}
\end{proof}

\begin{proof}[Proof of Theorem~\ref{thm:onthefy}]
We start by proving that the algorithm always terminates. Denote by $R_i$, $\hat{R}_i$, $Q_i$, and $\hat{Q}_i$ respectively
the value of the variables $R$, $\hat{R}$, $Q$, and $\hat{Q}$ at the beginning of the $i$-th iteration of the while loop (lines~\ref{alg:whilebegin}--\ref{alg:whileend}).
By induction on $i \in \mathbb{N}$ we can prove that
\begin{align}
 &R_i \cup \hat{R}_i \subseteq R_{i+1} \cup \hat{R}_{i+1} \, , &&  Q_i \cup \hat{Q}_i \subseteq Q_{i+1} \cup \hat{Q}_{i+1} \\
 &\hat{R}_i \subseteq \hat{R}_{i+1} \, ,  &&  \hat{Q}_i \subseteq \hat{Q}_{i+1} \label{eq:hatchain} \\
 &R_i \cap \hat{R}_i = \emptyset \, , &&  Q_i \cap \hat{Q}_i = \emptyset \,. \label{eq:disjointchain}
\end{align}
Since $R_i, \hat{R}_i \subseteq \X \times \X$ and $\X$ is finite, the increasing chains
$\{R_i \cup \hat{R}_i\}_{i \in \mathbb{N}}$ and $\{\hat{R}_i \}_{i \in \mathbb{N}}$ are finite.
Let $n \in \mathbb{N}$ be an index that is a limit index for both chains, then the following equalities hold:
\begin{align*}
	R_n &=  (R_n \cup \hat{R}_n) \setminus \hat{R}_n \tag{$R_n \cap \hat{R}_n = \emptyset$} \\
	&= (R_{n+1} \cup \hat{R}_{n+1}) \setminus \hat{R}_n \\
	&= (R_{n+1} \cup \hat{R}_{n+1}) \setminus \hat{R}_{n+1} \\
	&= R_{n+1} \tag{$R_{n+1} \cap \hat{R}_{n+1} = \emptyset$}
\end{align*}
Let $\mon_f$ denote the set of monomials that occur in the vector field $f$. Clearly, $\mon_f$ is finite.
Note that for all $x,y \in \X$ and $\omega \in \Gamma_\lin(f_x,f_y)$, $\supp(\omega) \subseteq \mon_f \times \mon_f$.
Therefore $Q_i, \hat{Q}_i \subseteq \mon_f \times \mon_f$ for all $i \in \mathbb{N}$. By following a similar argument as
before, we can prove that for some $n \in \mathbb{N}$, $Q_n = Q_{n+1}$.
Since the condition of the while loop checks at each iteration $i$ if $(R_i,\hat{R}_i, Q_i,\hat{Q}_i) \neq (R_{i+1}, \hat{R}_{i+1},Q_{i+1}, \hat{Q}_{i+1})$, we have that after finitely many iterations, the condition is falsified. Hence the algorithm terminates.

Let $\mathcal{R}$ be the output of $\onthefly(f,\id{Query},C)$. As explained above, during the last iteration of the while-loop neither $R$ nor $Q$ are changed w.r.t.\ the previous iteration. Hence, during the last iteration of the while-loop, each iteration of the for-loop on $R$ (lines~\ref{alg:loopRbegin}--\ref{alg:loopRend}) as well as each iteration of the for-loop on $Q$ (lines~\ref{alg:loopQbegin}--\ref{alg:loopQend}) executes the first branch of the if statement.
Therefore, for all $(m,n) \in Q$ there exists $\rho \in \Gamma_\mon(m,n)$ such that $R = R \cup \supp(\rho)$. Hence
\begin{align*}
Q &\subseteq \{(m,n) \mid \exists \rho \in \Gamma_\mon(m,n).\, R = R \cup \supp(\rho) \} \\
&= \{(m,n) \mid \exists \rho \in \Gamma_\mon(m,n).\, \supp(\rho) \subseteq R \} \\
 &= \mon[R]
\end{align*}
Analogously, for all $(x,y) \in R$ exists $\omega \in \Gamma_\lin(f_x,f_y)$ such that $Q = Q \cup \supp(\omega)$.
Hence,
\begin{align*}
R &\subseteq \{(x,y) \mid \exists \omega \in \Gamma_{\lin}(f_x,f_y) . \, Q = Q \cup \supp(\omega) \} \\
&\subseteq \{(x,y) \mid \exists \omega \in \Gamma_{\lin}(f_x,f_y) . \, \supp(\omega) \subseteq Q \} \\
&= \{ (x,y) \mid (f_x,f_y) \in \lin[Q] \} \\
&\subseteq \{ (x,y) \mid (f_x,f_y) \in \lin[\mon[R]] \} \tag{$Q \subseteq \mon[R]$}\\
&=  \mathcal{B}(R) \tag{def.\ $\mathcal{B}$}
\end{align*}
This proves that $R$ is a BDB. By $\hat{R}_0 = C$, \eqref{eq:hatchain}
 and \eqref{eq:disjointchain}, we have $R \cap C = \emptyset$. Therefore, by Theorem~\ref{thm:constrBR}, $R$ is a $C$-constrained BDB.


We show that for any $i \in \mathbb{N}$
\begin{equation}
\mon[(\X\times \X) \setminus \hat{R}_i] \subseteq (\mon \times \mon) \setminus \hat{Q}_i
\label{wrongmonomialpairs}
\end{equation}
For $i = 0$ the above inclusion holds because $Q_0 = \emptyset$. Assume
towards a contradiction that there exists $i > 0$ such that $(m, n) \in \mon[(\X\times \X) \setminus \hat{R}_i]$ and $(m, n) \in \hat{Q}_i$.
By construction of the algorithm, if $(m, n) \in \hat{Q}_i$ then at some iteration $j < i$ of the while-loop, the pair $(m, n)$ was moved from
$Q$ to $\hat{Q}$ by executing line~\ref{alg:movetoQhat}. This means that the condition of the if statement in line~\ref{alg:ifRho} was false, that is
$\supp(\rho) \cap \hat{R}_j \neq \emptyset$ for all $\rho \in \Gamma_\mon(m,n)$.
This implies also that $\supp(\rho) \cap \hat{R}_i \neq \emptyset$ because $j < i$ implies $\hat{R}_j \subseteq \hat{R}_i$.
This contradicts the fact that $(m, n) \in \mon[(\X\times \X) \setminus \hat{R}_i]$ because
\begin{align*}
	&\mon[(\X\times \X) \setminus \hat{R}_i] \\
	&= \{ (m,n) \mid \exists \rho \in \Gamma_\mon(m,n) .\, \supp(\rho) \subseteq (\X\times \X) \setminus \hat{R}_i \} \\
	&=  \{ (m,n) \mid \exists \rho \in \Gamma_\mon(m,n) .\, \supp(\rho) \cap \hat{R}_i = \emptyset \} \,.
\end{align*}
Now we show by induction on $i \in \mathbb{N}$ that
\begin{equation}
\mathcal{B}_C^{i}((\X \times \X) \setminus C) \subseteq (\X \times \X) \setminus \hat{R}_i
\label{wrongspecies}
\end{equation}
The base case ($i = 0$) holds because $\hat{R}_0 = C$.
For the inductive step, consider $i \geq 0$, then we have
\begin{align*}
	\mathcal{B}_C^{i+1}((\X \times \X) \setminus C) & = \mathcal{B}_C(\mathcal{B}_C^{i}((\X \times \X) \setminus C)) \\
	&\subseteq \mathcal{B}_C((\X \times \X) \setminus \hat{R}_i)
	\tag{ind.\ hp. and $\cal{B}_C$ monotone} \\
	&\subseteq \mathcal{B}((\X \times \X) \setminus \hat{R}_i) \setminus C
	\tag{Eq.~\eqref{eq:BC} and $C \subseteq \hat{R}_i$} \\
	&\subseteq \{ (x,y) \mid (f_x, f_y) \in \lin[\mon[(\X \times \X) \setminus \hat{R}_i]] \}  \tag{def.\ $\mathcal{B}$}\\
	&\subseteq \{ (x,y) \mid (f_x, f_y) \in \lin[(\mon \times \mon) \setminus \hat{Q}_i] \} \tag{Eq.\ \eqref{wrongmonomialpairs}}\\
	&= \{ (x,y) \mid \exists \omega \in \Gamma_\lin(f_x, f_y). \, \supp(\omega) \cap \hat{Q}_i = \emptyset \} \\
	&\subseteq (\X \times \X) \setminus \hat{R}_{i+1}  \,.
\end{align*}

Then,
\begin{align*}
\mathtt{gfp}(\mathcal{B}_C) &= \textstyle \inf_{i \in \mathbb{N}} \mathcal{B}_C^{i}(\X\times \X) \tag{Kleene fixed-point thm.}\\
& \subseteq \textstyle \inf_{i \in \mathbb{N}} \big( (\X \times \X) \setminus \hat{R}_i \big) \tag{Eq.\ \eqref{wrongspecies} and Eq.~\eqref{eq:BC}} \\
&=  (\X \times \X) \setminus \textstyle \sup_{i \in \mathbb{N}} \hat{R}_i \,.
\end{align*}
Therefore, the value of $\hat{R}$ at the end of the while-loop satisfies $\mathtt{gfp}(\mathcal{B}_C) \cap \hat{R} = \emptyset$.
Since $R_0 = \id{Query} \setminus C$, $\hat{R}_0 = C$ and the chain $\{R_i \cup \hat{R}_i\}_{i \in \mathbb{N}}$ is increasing, at the end of the while-loop
any pair $(x,y) \in \id{Query}$ either belongs to $R$ or $\hat{R}$. Since $R$ is a $C$-constrained BDB we have
that $(x,y) \in R$ implies $(x,y) \in \mathtt{gfp}(\mathcal{B}_C)$. Conversely, if $(x,y) \not\in R$, then $(x,y) \in \hat{R}$, which in turn implies that $(x,y) \not\in \mathtt{gfp}(\mathcal{B}_C)$.
\end{proof}

\begin{proof}[Proof of Theorem~\ref{thm:alg:otf:comp}]
Assume that each polynomial $f_{x_i}$ in $f$ is of the form $f_{x_i} = \sum_{l = 1}^{L_i} \alpha_{i,l} m_{i,l}$ for each $x_i \in \X$. Let $k$ be the maximum number of monomials occurring in each polynomial, that is $k = \max_{i} L_i$; and let $h$ the the greatest number of variables occurring in each monomial $m_{i,l}$.
Using Orlin's algorithm~\cite{Olrin88} one can find a linear coupling $\omega$ satisfying the condition of line~\ref{alg:ifOmega} in time $\mathcal{O}(k + k \log k)$, by solving an uncapacitated minimum cost flow problem. Analogously, finding the monomial coupling satisfying the condition of line~\ref{alg:ifRho} takes $\mathcal{O}(h + h \log h)$.

We assume $R$ and $\hat{R}$ to be implemented as two $|\X| \times |\X|$ boolean matrices, while $Q$ and $\hat{Q}$ are assumed to be implemented as two $|\mon_f| \times |\mon_f|$ boolean matrices.
With this in place, executing lines~\ref{alg:init}--\ref{alg:initend} takes time $\mathcal{O}(|\X|^2)$.

A single execution of the for-loop in lines~\ref{alg:loopRbegin}--\ref{alg:loopRend} iterates at most $|\X|^2$ times. As said before executing line~\ref{alg:ifOmega} takes $\mathcal{O}(k + k \log k)$, while executing line~\ref{alg:up-to-Mg}  takes $\mathcal{O}(k^2)$ since $|\supp(\omega)| \leq k^2$; and executing line~\ref{alg:moveR} takes constant time. Overall, one execution of lines~\ref{alg:loopRbegin}--\ref{alg:loopRend} takes
$\mathcal{O}(|\X|^2 (k^2 + k + k \log k) ) = \mathcal{O}(|\X|^2 k^2)$.

Analogously, one can show that a single execution of the for-loop in lines~\ref{alg:loopQbegin}--\ref{alg:loopQend} takes $\mathcal{O}(|\mon_f|^2 (h^2 + h +  h \log h)) = \mathcal{O}(|\mon_f|^2 h^2)$.

The number of iterations of the while-loop (lines~\ref{alg:whilebegin}--\ref{alg:whileend}) is bounded by $2 (|\X|^2 + |\mon_f|^2)$ because, as discussed in the proof of Theorem~\ref{thm:onthefy},
$\{ R_i \cup \hat{R}_i \}_{i \in \mathbb{N}}$ and $\{ \hat{R}_i \}_{i \in \mathbb{N}}$ are increasing chains bounded by $\X \times \X$, while $\{ Q_i \cup \hat{Q}_i \}_{i \in \mathbb{N}}$ and $\{ \hat{Q}_i \}_{i \in \mathbb{N}}$ are increasing chains bounded by $\mon_f \times \mon_f$.

Taking into account that $|\X| \leq |\mon_f|$, the time-complexity of Algorithm~\ref{alg:otf} simplifies to follows
\end{proof}

\begin{proof}[Proof of Theorem~\ref{thm:findFDB}]
Similar arguments used for Theorem~\ref{thm:onthefy}
\end{proof}

\begin{algorithm}[pt]
\begin{codebox}
\Procname{$\jacdec\big(\text{Polynomial vector field } f = (f_{x_i})_{x_i \in \X} \text{ with } f_{x_i} = \sum_{l=1}^{L_i} \alpha_{i,l} m_{i,l}\big)$} \li $M \gets \emptyset$
\li \For \textbf{each} $x_i, x_j \in \X$ \Do 
\li         \For \textbf{each} $1 \leq l \leq L_i$ \Do
    \li         \If $m_{i,l}(x_j) > 0$ \Then
    \li             $n \gets m_{i,l}$
    \li             $n(x_j) \gets m_{i,l}(x_j) - 1$
    \li         \If $n \notin M$ \Then
    \li             $M \gets M \cup \{n\}$
    \li             $J_n \gets \emptyset$
                \End
    \li             $J_n(x_i,x_j) = J_n(x_i,x_j) + m_{i,l}(x_j) \cdot \alpha_{i,l}$
                \End
    \End
\End \label{alg:loopRend}
\li \Return $\{ J_n \mid n \in M \}$
\end{codebox}
\caption{Computation of the Jacobi matrix decomposition~(\ref{eq:jac:dec}).}\label{alg:jac:dec}
\end{algorithm}

\begin{proof}[Proof of Lemma~\ref{lem:jac:dec}]
Since $\sum_{x_i \in \X} L_i = |\mon_f|$, there are at most $|\mon_f|$ pairwise different monomials in the vector field. This and the fact that each monomial gives rise to at most $|\X|$ further monomials by means of partial differentiation, implies that the complexity of Algorithm~\ref{alg:jac:dec} is bounded by $\mathcal{O}(|\mon_f| |\X|)$ (provided that matrices are stored as sparse matrices, i.e., lists). With this, both estimations follow by noting that a partial differentiation of a monomial $m_{i,l}$ either gives rise to a new matrix $J_n$ with exactly one non-zero entry or updates (possibly a non-zero) entry of a previously created matrix $J_n$.
\end{proof}

\begin{proof}[Proof of Theorem~\ref{thm:complexityFindFDB}]
Assume $R$ and $\hat{R}$ to be implemented as two $|\X| \times |\X|$ boolean matrices.
With this in place, executing lines~\ref{lin:FDBinit} takes time $\mathcal{O}(|\X|^2)$.
By Lemma~\ref{lem:jac:dec}, executing line~\ref{lin:jacdec} takes $\mathcal{O}(|\X||\mon_f|)$.

Using Orlin's algorithm~\cite{Olrin88} one can find a linear coupling $\omega_f$ satisfying the condition of line~\ref{alg:omegafk} in time $\mathcal{O}(|\X| + |\X| \log |\X|)$ by solving an uncapacitated minimum cost flow problem, and by Lemma~\ref{lem:jac:dec}, $\kappa \leq |\X||\mon_f|$. Therefore, a single execution of line~\ref{alg:omegafk} takes $\mathcal{O}(|\X|^2 |\mon_f| \log |\X|)$.
Executing line~\ref{lin:kR} takes $\mathcal{O}(|\X|^3 |\mon_f|)$, while line~\ref{lin:kmove} takes constant time.
Therefore, an execution of the for-loop in lines~\ref{lin:kforbegin}--\ref{lin:kforend} takes
$\mathcal{O}(|\X|^5 |\mon_f|)$ because $|R| \leq |X|^2$.

The number of iterations of the while-loop (lines~\ref{lin:whilefdb}--\ref{lin:kwhileend}) is bounded by $2 |\X|^2$ because, $\{ R_i \cup \hat{R}_i \}_{i \in \mathbb{N}}$ and $\{ \hat{R}_i \}_{i \in \mathbb{N}}$ are increasing chains bounded by $\X \times \X$.
Therefore, Algorithm~\ref{alg:findFBD} runs in time $\mathcal{O}(|\X|^7 |\mon_f|)$.
\end{proof}

\begin{proof}[Proof of Theorem~\ref{thm:uptoAlg}]
Termination and \ref{itm:2:thm:uptoAlg} follow analogously to~Thm.\ref{thm:onthefy}. Here we focus on \ref{itm:1:thm:uptoAlg}. Let $R$ be the output of $\onthefly(f,\id{Query},C,g)$. In the last iteration of the while-loop neither $R$ nor $Q$ are changed w.r.t.\ the previous iteration.
Hence, during the last iteration of the while-loop, each iteration of the for-loop on $R$ (lines~\ref{alg:loopRbegin}--\ref{alg:loopRend}) as well as each iteration of the for-loop on $Q$ (lines~\ref{alg:loopQbegin}--\ref{alg:loopQend}) executes the first branch of the if statement.

Therefore, for all $(m,n) \in Q$ there exists $\rho \in \Gamma_\mon(m,n)$ such that
$R = R \cup ( \supp(\rho) \setminus g(R))$. The following implications hold
\begin{align*}
&R = R \cup ( \supp(\rho) \setminus g(R)) \\
&\implies \supp(\rho) \setminus g(R) \subseteq R \\
&\implies  \supp(\rho) \setminus (g(R) \setminus C) \subseteq R \tag{$\supp(\rho) \cap C = \emptyset$}\\
&\implies \supp(\rho) \setminus (g(R) \setminus C) \subseteq R \setminus C \tag{$R \cap C = \emptyset$} \\
&\implies \supp(\rho) \setminus (g(R) \setminus C) \subseteq (g(R) \setminus C) \tag{$R \subseteq g(R)$} \\
&\implies \supp(\rho) \subseteq g(R) \setminus C \,.
\end{align*}
Therefore $Q \subseteq \mon[g(R) \setminus C]$.

Analogously, for all $(x,y) \in R$ there exists $\omega \in \Gamma_\lin(f_x,f_y)$ such that $Q = Q \cup (\supp(\omega) \setminus \mon[g(R) \setminus C])$.
The following implications hold
\begin{align*}
&Q = Q \cup (\supp(\omega) \setminus \mon[g(R) \setminus C]) \\
&\implies \supp(\omega) \setminus \mon[g(R) \setminus C] \subseteq Q \\
&\implies \supp(\omega) \setminus \mon[g(R) \setminus C] \subseteq \mon[g(R) \setminus C] \\ 
&\implies \supp(\omega) \subseteq \mon[g(R) \setminus C] \,.
\end{align*}

Form this we have
\begin{align*}
R 
&\subseteq \{(x,y) \mid \exists \omega \in \Gamma_{\lin}(f_x,f_y) . \, \supp(\omega) \subseteq \mon[g(R) \setminus C] \} \\
&= \{ (x,y) \mid (f_x,f_y) \in \lin[\mon[g(R) \setminus C]] \} \\
&=  \mathcal{B}(g(R) \setminus C) \tag{def.\ $\mathcal{B}$}
\end{align*}
Since $R \subseteq \mathcal{B}(g(R) \setminus C)$ and $R \cap C = \emptyset$
we have $$\mathcal{B}(g(R) \setminus C) = \mathcal{B}(g(R) \setminus C) \setminus C = \mathcal{B}_C(g(R)) \,.$$ Therefore, $R$ is a $(\mathcal{B}_C \circ g)$-simulation.
\end{proof}

\begin{proof}[Proof of Lemma~\ref{lem:B-compatible}]
Note that an extensive function $g$ that satisfies
\begin{equation}
g \circ \mathcal{B} \circ g = \mathcal{B} \circ g, \label{eq:preservesg}
\end{equation}
is also $\mathcal{B}$-compatible, because
\begin{align*}
g \circ \mathcal{B} &\subseteq g \circ \mathcal{B} \circ g
\tag{$g \circ \mathcal{B}$ monotone, $g$ extensive} \\
&= \mathcal{B} \circ g \tag{by \eqref{eq:preservesg} }\,.
\end{align*}

Since $r$, $s$, $t$, and $e$ are extensive functions, we are only left to prove that each satisfy \eqref{eq:preservesg}. By definition of $\cal{B}$, this corresponds to checking that for $g \in \{r, s, t, e\}$ it holds
$\poly[g(\R)] = g(\poly[g(\R)])$. The property holds trivially for $r$ and $s$ by definition of linear and monomial couplings. $\poly[e(\R)] = e(\poly[e(\R)])$, follow by Theorem~\ref{thm:monStrassen} and Theorem~\ref{thm:linStrassen}, which can be respectively used to prove that the monomial lifting and the linear lifting of an equivalence relation are equivalence relations.

Consider now $t$. We start proving that the monomial lifting of an a transitive relation is transitive, i.e., $t(\mon[t(\R)]) = \mon[t(\R])]$. Assume $(m,n),(n,o) \in \mon[t(\R])]$, we will show that $(m,o) \in \mon[t(\R])]$. For this, it will be convenient to highlight that the existence of a monomial coupling $\rho_{mn} \in \Gamma_{\mon}(m,n)$ corresponds one-to-one to the existence of a feasible network flow on the bipartite directed graph $G(m,n) = (\X_m \uplus \X_n ,E(m,n))$ with
\begin{gather*}
\X_m = \{ x_m \mid x \in \X \} \, , \quad
\X_n = \{ y_n \mid y \in \X \} \,, \\
E(m,n) = \{(x_m, y_n) \mid (x,y) \in \supp(\rho_{mn})\},
\end{gather*}
where each node supply/demand is defined as $b(x_m) = m(x)$ for $x_m \in \X_m$, and $b(y_n) = - n(y)$ for $y_n \in \X_n$.

With this in place, if there exist $\rho_{mn} \in \Gamma_{\mon}(m,n)$ and $\rho_{no} \in \Gamma_{\mon}(n,o)$ such that $\supp(\rho_{mn}), \supp(\rho_{no}) \subseteq t(\R])$, one can compose the two networks arising from the couplings ensuring that there exists a feasible network flow for the directed graph $G(m,n,o) = (\X_m \uplus \X_n \uplus \X_o, E(m,n,o))$ where $E(m,n,o) = E(m,n) \cup E(n,o)$ and each node supply/demand is
defined as $b(x_m) = m(x)$ for $x_m \in \X_m$, $b(y_n) = 0$ for $y_n \in \X_n$, and $b(z_o) = -o(z)$ for $z_o \in \X_0$. Intuitively, the nodes $\X_n$ become ``transhipment" nodes, while the nodes in $\X_m$ and $\X_o$ will be respectively source and target nodes.

By removing all ``transhipment" nodes from the above network and connecting source nodes to reachable target nodes, we obtain a network flow with graph $G(m,o) = (\X_m  \uplus \X_o, E(m,o))$ where
$$E(m,o) {=} \{ (x_m, z_o) \mid (x_m, y_n) \in E(m,n) \land (y_n, z_o) \in E(n,o) \}$$
and with each node supply/demand defined as $b(x_m) = m(x)$ for $x_m \in \X_m$, and $b(z_o) = - o(z)$ for $z_o \in \X_o$. The existence of a feasible flow for the latter network is ensured by the existence of a feasible flow for the composite network. Therefore, $(m,o) \in \mon[t(\R)]$.
Analogously, we can prove that $\poly[t(\R)] = \lin[\mon[t(\R)]]$ is a transitive relation, leveraging on the fact that the know that $\mon[t(\R)]$ is a transitive relation.
\end{proof}

\begin{proof}[Proof of Lemma~\ref{lem:reflupto}]
As noted in Lemma~\ref{lem:B-compatible} it suffice to show $r \circ \mathcal{B}_{C} \circ r = \mathcal{B}_{C} \circ r$. We prove the two inclusions separately.
($\supseteq$) Holds true because $r$ is extensive.
($\subseteq$) Let $(x,y) \in r (\mathcal{B}_{C}(r(\R)))$. Since $r (\mathcal{B}_{C}(r(\R))) = \mathtt{id} \cup \mathcal{B}_{C}(r(\R))$, two cases are possible. If $(x,y) \in \mathcal{B}_{C}(r(\R))$, then we are done. Otherwise we have that $x = y$. Note that the identity relation $\mathtt{id}$ is an equivalence therefore, by Corollary~\ref{cor:couplingEquiv}, $(f_x,f_x) \in \poly[\mathtt{id}]$ for all $x \in \X$. With this in mind, the following chain of implications hold
\begin{align*}
(f_x,f_x) \in \poly[\mathtt{id}]
&\implies (f_x,f_x) \in \poly[\mathtt{id} \cup (\R \setminus C)] \tag{$\poly$ monotone} \\
&\implies (f_x,f_x) \in \poly[r(\R) \setminus C] \tag{$r(\R) = \mathtt{id} \cup \R$ and $C \cap \mathtt{id} = \emptyset$}
\end{align*}
Since $C \cap \mathtt{id} = \emptyset$, by definition of $\mathcal{B}_C$, we conclude that $(x,x) \in \mathcal{B}_C(r(\R))$.
\end{proof}

\begin{proof}[Proof of Lemma~\ref{lem:BsC-compatible}]
As noted in Lemma~\ref{lem:B-compatible} it suffice to show $s \circ \mathcal{B}_{s(C)} \circ s = \mathcal{B}_{s(C)} \circ s$. We prove the two inclusions separately.
($\supseteq$) Holds true because $s$ is extensive.

($\subseteq$) Let $(x,y) \in s (\mathcal{B}_{s(C)}(s(\R)))$. If $(x,y) \in \mathcal{B}_{s(C)}(s(\R))$, then we are done. Otherwise, we have $(y,x) \in \mathcal{B}_{s(C)}(s(\R))$, which implies
$s(\{ (x,y) \}) \cap s(C) = \emptyset$ and $(f_y, f_x) \in \poly[s(\R)\setminus s(C)]$.

Therefore, there exist $\omega_{yx} \in \Gamma_\lin(f_y, f_x)$ such that $\supp(\omega) \subseteq \mon[s(\R) \setminus s(C)]$, and for all $(n,m) \in \supp(\omega_{yx})$, there exists $\rho_{nm} \in \Gamma_{\mon}(n,m)$ such that $\supp(\rho_{nm}) \subseteq s(\R) \setminus s(C)$.

Define $\omega_{xy} \in \Gamma_\lin(f_x, f_y)$ as $\omega_{xy}(m,n) = \omega_{yx}(n,m)$ for all $n,m \in \mon$. Then, for $(m,n) \in \supp(\omega_{x,y})$, define $\rho_{mn} \in \Gamma_{\mon}(m,n)$ as $\rho_{mn}(x,y) = \rho_{nm}(y,x)$ for all $x,y \in \X$.
It is easy to see that if $(m,n) \in \supp(\omega_{xy})$, then $\supp(\rho_{mn}) \subseteq s(\R) \setminus s(C)$, i.e., $(m,n) \in \mon[s(\R) \setminus s(C)]$. Hence, $(f_x, f_y) \in \poly[s(\R) \setminus s(C)]$.
Since $(x,y) \not\in s(C)$, then $(x,y) \not\in C$. Therefore, $(x,y) \in s (\mathcal{B}_{s(C)}(s(\R)))$.
%
\end{proof}

\end{document}